\documentclass[superscriptaddress,nofootinbib]{revtex4}
\usepackage{amsmath, amsthm, amssymb}
\usepackage{braket}
\usepackage{graphicx}
\usepackage{varwidth}
\usepackage{tikz}
\usepackage{subfigure}
\usepackage{quantikz} 
\usepackage[
    colorlinks=true,    
    linkcolor=blue,     
    citecolor=blue,     
    urlcolor=blue       
]{hyperref}

\newtheorem{theorem}{Theorem}[section]

\newtheorem{dfn}[theorem]{Definition}
\newtheorem{example}[theorem]{Example}

\begin{document}

\title{Alleviating Post-Linearization Challenges for Solving Nonlinear Systems on a Quantum Computer}

\author{Tayyab Ali}\thanks{alitayyab.pu@gmail.com}
\affiliation{Canadian Quantum Research Center, 460 Doyle Ave 106, Kelowna, BC V1Y 0C2, Canada}

\begin{abstract}

The linearity inherent in quantum mechanics limits current quantum hardware from directly solving nonlinear systems governed by nonlinear differential equations. One can opt for linearization frameworks such as Carleman linearization, which provides a high dimensional infinite linear system corresponding to a finite nonlinear system, as an indirect way of solving nonlinear systems using current quantum computers. We provide an efficient data access model to load this infinite linear representation of the nonlinear system, upto truncation order $N$, on a quantum computer by decomposing the Hamiltonian into the weighted sum of non-unitary operators, namely the Sigma basis. We have shown that the Sigma basis provides an exponential reduction in the number of decomposition terms compared to the traditional decomposition, which is usually done in a linear combination of Pauli operators.  Once the Hamiltonian is decomposed, we then use the concept of unitary completion to construct the circuit for the implementation of each weighted tensor product component $\mathcal{H}_{j}$ of the decomposition.
\end{abstract}

\maketitle

\section{Introduction}
In both classical and quantum computations, differential equations play a central role in the modeling of physical systems to understand complex relationships using mathematical representations. There are various contexts where differential equations are ubiquitous in various sectors, including science and technology, financial engineering\cite{duffy2013finite}, weather modeling and atmospheric dynamics\cite{steppeler2024short}, healthcare and business. However, analytical solutions to these equations are not always possible, so we rely on discrete models and direct numerical techniques\cite{butcher2016numerical,boyd2001chebyshev}. Numerical methods often require recursive calculations to obtain an approximation which brings us closer to the true solution of the system. It becomes more difficult for us to converge our solution to the actual one if we are dealing with nonlinear or higher-dimensional dynamical systems. Case in point, in physics, plasma is composed of a large number of charged and neutral particles, and the evolution of such a multi-species environment is governed by highly nonlinear partial differential equations. Although we have computing as a tool to tackle this problem, it also limits when it comes to computational time and cost. In addition, unlike linear systems, nonlinear systems, whose dynamics is governed by nonlinear differential equations, usually do not have closed-form (analytical) solution. These systems are very sensitive to initial conditions, and solving them numerically could result in error amplification.

The criteria for testing a computing device, whether it works on classical or quantum principles, lies in the efficiency of algorithm and accuracy of the solution. If it solves the problem efficiently and the proposed solution is not contradictory to the true solution, then we are good to go. Another test is computational resources, cost, and time. This is the point where classical computers are lacking. Current classical computers are struggling with large datasets and high-dimensional problems. Among many, two prime reasons for that are hardware limitations of classical machines as they cannot offer enough capacity of memory units for certain problems, those are intrinsically complex, and algorithmic bottlenecks due to exponential scaling of problem complexity. Even if we get a large storage capacity,  classical algorithms would still take an exponential amount of time to achieve the solution. This is because they use bits and binary logic for calculations, which limits the clock speed. Since most of the high-dimensional problems are nonlinear, we need to model complex interaction from particle to particle, making it impractical to solve on a classical computer. Moreover, we cannot simulate the quantum mechanical description of the system on a classical computer\cite{feynman2018simulating}. Consider, we are exploring many body quantum phenomena in which we aim to simulate the influence of spin of one particle to another. Since quantum particles are in superposition, we need a two-dimensional vector to store the amplitude of each particle of being in $\ket{0}$ and $\ket{1}$, that is, for 50 particles we need to store $\sim 2^{50}$ many amplitudes, which is beyond the capabilities of the classical memory storage. In a nutshell, we need to look beyond the classical infrastructure, where we could deal with energy cost and time, that is, we must let our computing devices be built on quantum-mechanical elements.

Quantum scientific computing has now moved from vision to real-word impact, where it can solve various computational tasks in feasible amounts of time and computational cost compared to processors that are bound to follow the laws of classical physics. The quantum advantage in polynomial and exponential speedup has been demonstrated in certain problems, including factoring\cite{shor1994proceedings} and unstructured search\cite{grover1996proceedings}. In \cite{arute2019quantum}, the authors demonstrated the ability of a quantum processor to outperform the classical supercomputer with exponential speedup. They performed random circuit sampling on a superconducting quantum processor, consisting of 53 qubits, in $200 s$ that a classical supercomputer would take $10,000\,\,years$. However, one cannot conclude that quantum processors are universally fast in all tasks; instead, they are capable of showing a quantum advantage in selected problems only. 

It is extraordinarily valuable to note that the most natural application of quantum computing is to solve linear systems\cite{harrow2009quantum,subacsi2019quantum,berry2014high,berry2017quantum,childs2017quantum,childs2020quantum,costa2022optimal}, since quantum mechanics is itself described by linear differential equations. However, most of the real-world systems are inherently nonlinear, and their dynamics is governed by nonlinear differential equations. For that, we need to develop quantum algorithms for nonlinear differential equations in order to integrate the quantum advantage in nonlinear dynamical systems or introduce nonlinear corrections into ordinary quantum mechanics, as nonlinear variants can solve problems that are difficult to solve with standard quantum theory\cite{rowinski2019foundations,childs2016optimal,abrams1998nonlinear}.

To address this problem, several attempts have been made to develop a quantum algorithm that can efficiently solve nonlinear dynamical systems. Since we are in the NISQ (Noisy Intermediate-Scale Quantum Computing) era, the best advantage we could take from our computing devices is to let them do what they are at best; that is, we need to develop hybrid classical-quantum algorithms. Such approaches usually utilize classical optimization for finding the best variational parameters, while a quantum computer is used to prepare parametrized quantum states for measuring the ground state approximation. Regarding our problem, one of these algorithms is given in \cite{lubasch2020variational}, where the authors have transformed a nonlinear differential equation into a minimization problem. Using the finite difference method, they defined a cost function as a sum of linear and nonlinear interactions,
$$
\mathcal{C} = <K> + <P> + <I>
$$
where K, P, and I represent the terms of kinetic energy, potential, and interaction, respectively. The idea is to get ground state approximation using minimization of this cost function and for that they have computed these terms, especially the nonlinear term, using a quantum nonlinear processing unit (QNPU) which takes different forms on a gate-based quantum hardware depending upon which term we are going to evaluate from our Hamiltonian. Since the accuracy of the approximation depends on the trial solution, they tested the algorithm for the variational quantum ansatz and matrix product state, from which the former proved advantageous over the latter. For nonlinear interaction, they prepared multiple copies of the trial solution and let them interact with an ancilla qubit using entangled operations. Although the method can be applied to a broad range of nonlinear terms, it requires high connectivity to implement non-local interactions.

Another hybrid approach was given by \cite{kyriienko2021solving}. This approach uses a quantum feature map to encode the variables as rotation angles of qubits, which naturally produces the Chebyshev polynomials that are used for complex function approximations. Despite not having a guaranteed computational cost and convergence after the optimization loop, this algorithm works for small instances. The authors defined the functions as expectation values of parametrized quantum circuits instead of directly encoding them in the amplitudes of a quantum state. In this way, they handled high-dimensional inputs in a more flexible way than \cite{lubasch2020variational}, because of which the circuit depth is decreased and they become more hardware efficient.

In addition to these hybrid algorithms, linearization frameworks have also been studied to find approximate global linear representations of nonlinear systems\cite{joseph2020koopman,tennie2025quantum}, so that we can solve them using existing quantum algorithms for linear systems. One such technique is \textit{Carleman linearization}, which takes a finite-dimensional nonlinear dynamical system and converts it into a linear system with infinite degrees of freedom\cite{vaszary2024carleman,kowalski1991nonlinear}. In the context of quantum computing, this work has been done in \cite{liu2021efficient}, where the authors provided a high-level theoretical description of this technique for dissipative nonlinear ODEs. However, especially in this case, linearization does not guaranty an immediate solution to the problem using existing quantum algorithms. After linearization, the system will have infinite degrees of freedom and one could lose all the quantum advantage by just loading this much data from a classical computer to a quantum computer. For this reason, efficient data access models should be discussed. Additionally, after loading, optimizing this large amount of data is also a central issue. Since they did not provide a practical implementation of their algorithm, their paper did not take into account these possible problems. In this paper, we provide a comprehensive discussion of the key challenges that may arise when using linearization workflow to solve nonlinear systems on a quantum computer, along with potential solutions.

The rest of the paper is organized as follows. For the sake of convenience and to better capture the post-linearization challenges, we describe Carleman linearization in Section. \ref{section-2}. We have used a quadratic model of Bernoulli's equation to illustrate the convergence of the solution with an increase in the truncation order $N$. In Section. \ref{section-3} we present our main results. In Sec. \ref{section-3.1}, we have shown how traditional LCU in Pauli operators work and investigated the quadratic growth of decomposition terms over matrix size. We then define a set of non-unitary operators, Sigma basis, in Sec. \ref{section-3.2} and prove that the decomposition terms grow linearly with the non-zero entries of the matrix. By decomposing the Hamiltonian into this new basis set, we imply the idea of unitary completion to construct circuit for the tensor product of the components of Sigma basis in Section. \ref{section-3.2.1}. Section. \ref{section-4} investigates how optimization can serve as our next problem in solving the truncated linear system using variational quantum algorithms. Finally, we conclude our work in Section. \ref{section-5} by describing some potential areas where this work can be applicable.

\section{Carleman Linearization}\label{section-2}
To analyze the subsequent difficulties or post-linearization challenges, it is essential to have a clear understanding of the underlying technique. In this section, we provide a general overview for the linearization of a $p^{th}$ order vector-valued ordinary differential equation using Carleman embedding. We obtain a large sparse structured matrix corresponding to the linear representation of the system using truncation.
\subsection{Nonlinear Vector-Valued Ordinary Differential Equations}
We consider a nonlinear dynamical system with multiple state variables evolving over time $t$. The complete description of the state of the system is described by an n-dimensional state vector, $\boldsymbol{\Phi} = [\phi_{1}, \phi_{2},\cdots,\phi_{n}]^{T} \in \mathbb{R}^{n}$, comprising all state variables. Taking into account both dependent and independent state contributions, the system will take the form
\begin{equation}\label{eq:five}
    \frac{d \boldsymbol{\Phi}}{dt} = \mathcal{M}_{0} + \mathcal{M}_{1} \boldsymbol{\Phi} + \sum_{k=2}^{p} \mathcal{M}_{k} \boldsymbol{\Phi}^{\otimes k}, \,\,\,\,\,\,\, \boldsymbol{\Phi}(0) = \boldsymbol{\Phi}_{in}
\end{equation}
where $\mathcal{M}_{0} \in \mathbb{R}^{n}$ is the constant vector representing state-independent effects such as driven force, $\mathcal{M}_{1} \in \mathbb{R}^{n \times n}$ is linearly coupled to the state vector showing linear contributions, and $\mathcal{M}_{k} \in \mathbb{R}^{n\times n^{k}}$ is responsible for describing nonlinear couplings of order $k$ to the state vector. The k-fold Kronecker power  $\boldsymbol{\Phi}^{\otimes k} \in \mathbb{R}^{n^{k}}$ comprises all monomials of $\boldsymbol{\Phi}(t)$.
The heart of the technique lies in the definition of an auxiliary state to couple a finite-dimensional nonlinear system into an infinite-dimensional linear system in the new state variables. We define 
\begin{equation}\label{def;aux-vector}
    \boldsymbol{y}_{j}:= \boldsymbol{\Phi}^{\otimes j}
\end{equation} 
With this definition,
\begin{equation}\label{eq:six}
    \frac{d \boldsymbol{y}_{j}}{dt} = \frac{d}{dt} \left(\underbrace{\boldsymbol{\Phi} \otimes \boldsymbol{\Phi} \otimes \cdots \boldsymbol{\Phi}}_{j\,times}\right) = \sum_{i=1}^{j} \boldsymbol{\Phi}^{\otimes (i-1)} \otimes \dot{\boldsymbol{\Phi}} \otimes \boldsymbol{\Phi}^{\otimes(j-i)} 
\end{equation}
substituting $\dot{\boldsymbol{\Phi}}$ from Eqn.(\ref{eq:five}) in Eqn.(\ref{eq:six}),
$$
 = \sum_{i=1}^{j} \boldsymbol{\Phi}^{\otimes (i-1)} \otimes \left(\mathcal{M}_{0} + \mathcal{M}_{1}\boldsymbol{\Phi} + \sum_{k=2}^{p} \mathcal{M}_{k} \boldsymbol{\Phi}^{\otimes k}\right) \otimes \boldsymbol{\Phi}^{\otimes(j-i)} 
$$
$$
= \sum_{i=1}^{j} \boldsymbol{\Phi}^{\otimes (i-1)} \otimes \mathcal{M}_{0} \otimes \boldsymbol{\Phi}^{\otimes(j-i)} + \sum_{i=1}^{j} \boldsymbol{\Phi}^{\otimes (i-1)} \otimes \mathcal{M}_{1}\boldsymbol{\Phi} \otimes \boldsymbol{\Phi}^{\otimes(j-i)} + \sum_{k=2}^{p}\sum_{i=1}^{j} \boldsymbol{\Phi}^{\otimes (i-1)} \otimes \mathcal{M}_{k} \boldsymbol{\Phi}^{\otimes k} \otimes \boldsymbol{\Phi}^{\otimes(j-i)}
$$
\begin{equation}
\begin{split}
    =\left[\sum_{i=1}^{j} \mathbb{I}^{\otimes (i-1)} \otimes \mathcal{M}_{0} \otimes \mathbb{I}^{\otimes(j-i)}\right]\boldsymbol{\Phi}^{\otimes (i-1)} \otimes \boldsymbol{\Phi}^{\otimes(j-i)} + \left[\sum_{i=1}^{j} \mathbb{I}^{\otimes (i-1)} \otimes \mathcal{M}_{1}  \otimes \mathbb{I}^{\otimes(j-i)} \right]\boldsymbol{\Phi}^{\otimes (i-1)} \otimes \boldsymbol{\Phi} \otimes \boldsymbol{\Phi}^{\otimes(j-i)} \\ + \left[\sum_{k=2}^{p}\sum_{i=1}^{j} \mathbb{I}^{\otimes (i-1)} \otimes \mathcal{M}_{k} \otimes \mathbb{I}^{\otimes(j-i)}\right]\boldsymbol{\Phi}^{\otimes (i-1)} \otimes \boldsymbol{\Phi}^{\otimes k} \otimes \boldsymbol{\Phi}^{\otimes(j-i)}
\end{split}
\end{equation}
we can now represent our system in terms of $\boldsymbol{y}$\footnote{$\underbrace{\boldsymbol{\Phi}^{\otimes (i-1)}}_{(i-1)\,terms} \otimes \underbrace{\boldsymbol{\Phi}^{\otimes(j-i)}}_{(j-i)\, terms} = \underbrace{\boldsymbol{\Phi}^{\otimes (k-1)}}_{i - 1 + j-i = j-1}$}
$$
    =\underbrace{\left(\sum_{i=1}^{j} \mathbb{I}^{\otimes (i-1)} \otimes \mathcal{M}_{0} \otimes \mathbb{I}^{\otimes(j-i)}\right)}_{\mathcal{A}_{j-1}^{(0)}}\boldsymbol{y}_{j-1} + \underbrace{\left(\sum_{i=1}^{j} \mathbb{I}^{\otimes (i-1)} \otimes \mathcal{M}_{1}  \otimes \mathbb{I}^{\otimes(j-i)} \right)}_{\mathcal{A}_{j}^{(1)}}\boldsymbol{y}_{j}   + \sum_{k=2}^{p} \underbrace{\left(\sum_{i=1}^{j} \mathbb{I}^{\otimes (i-1)} \otimes \mathcal{M}_{k} \otimes \mathbb{I}^{\otimes(j-i)}\right)}_{\mathcal{A}_{k+j-1}^{(p)}}\boldsymbol{y}_{k+j-1}
$$
\begin{equation}\label{eq:eight}
    \frac{d \boldsymbol{y}_{j}}{dt} = \mathcal{A}_{j-1}^{(0)} \boldsymbol{y}_{j-1} + \mathcal{A}_{j}^{(1)} \boldsymbol{y}_{j} + \sum_{k=2}^{p} \mathcal{A}_{k+j-1}^{(k)} \boldsymbol{y}_{k+j-1} 
\end{equation}
or
\begin{equation}\label{eq:nine}
    \frac{d \boldsymbol{y}_{j}}{dt} = \sum_{k=0}^{p} \mathcal{A}_{k+j-1}^{(k)} \boldsymbol{y}_{k+j-1}, \,\,\,\, j = 1,2,3,\cdots
\end{equation}
where $p$ denotes the degree of nonlinearity. Thus, we have coupled our system with new variables so that the resulting equations are linear. To enable a simplified analysis or a practical use case of the technique, we need to project Eqn.(\ref{eq:nine}) onto a finite-dimensional subspace spanned by finite number of monomials; that is, we must truncate the system to some finite order. This will provide us with an approximate representation of our original nonlinear system that will be computationally feasible. It is quite obvious that the accuracy of the solution will improve as the truncation order increases. However, it becomes challenging to obtain numerical stability at large truncation order, as it requires a large number of computational resources for simulation. Hence, we must take into account that the truncation order is balanced such that we can solve the system up to required accuracy. We define a truncated state vector $\boldsymbol{y} = [\boldsymbol{y}_{1}^{T} \boldsymbol{y}_{2}^{T} \cdots \boldsymbol{y}_{N}^{T}]^{T}$, whose dimension can be determined using the hierarchical tensor expansion for our original variable $\boldsymbol{\Phi}$. In Eqn.(\ref{def;aux-vector}), each lifted variable $\boldsymbol{y}_{j}$ lies in a tensor space of dimension $d^{j}$. Hence, at the truncation order $N$, the total number of components is just the sum of all tensor dimensions up to order N,
\begin{equation}\label{dimension}
    \text{dim}(\boldsymbol{y})= \sum_{i=1}^{N} d^{N} = \frac{d(d^{N} - 1)}{d-1}
\end{equation}
With this truncation, we are now able to write our original nonlinear ODE as a finite and linearized system.
\begin{equation}\label{final-truncated-matrix}
\frac{d}{dt}\begin{pmatrix}
        \boldsymbol{y}_{1}\\[2pt]
        \boldsymbol{y}_{2}\\[2pt]
        \boldsymbol{y}_{3}\\[2pt]
        \vdots \\[2pt]
        \boldsymbol{y}_{N-1}\\[2pt]
        \boldsymbol{y}_{N}\\[2pt]
    \end{pmatrix}
 \begin{pmatrix} 
\mathcal{A}_{0}^{(0)} & \mathcal{A}_{1}^{(1)} & \mathcal{A}_{2}^{(2)} & \mathcal{A}_{3}^{(3)} & \cdots & \mathcal{A}_{p}^{(p)} &  &  & \\[2pt]
 & \mathcal{A}_{1}^{(0)} & \mathcal{A}_{2}^{(1)} & \mathcal{A}_{3}^{(2)} & \cdots & \mathcal{A}_{p}^{(p-1)} & \mathcal{A}_{p+1}^{(p)} &  & \\[2pt]
 &  &  & \ddots & \vdots & \vdots & \ddots & \\[2pt]
 &  &  & \mathcal{A}_{\,N-2}^{(0)} & \mathcal{A}_{\,N-1}^{(1)} & \mathcal{A}_{\,N}^{(2)} & \cdots & \mathcal{A}_{\,p+N-2}^{(p)}\\[2pt]
 &  & & & \vdots & \vdots & \ddots & \vdots\\[2pt]
 &  &  &  & \mathcal{A}_{N-1}^{(0)} & \mathcal{A}_{N}^{(1)} & \cdots & \mathcal{A}_{p+N-2}^{(p-1)} & \mathcal{A}_{N+p-1}^{(p)} 
\end{pmatrix} = 
\begin{pmatrix}
        \boldsymbol{y}_{1}\\[2pt]
        \boldsymbol{y}_{2}\\[2pt]
        \boldsymbol{y}_{3}\\[2pt]
        \vdots \\[2pt]
        \boldsymbol{y}_{N-1}\\[2pt]
        \boldsymbol{y}_{N}\\[2pt]
    \end{pmatrix}
\end{equation}
This approximation is now computationally feasible. However, discarding contributions beyond $N$ could result in an approximation and truncation errors. Although we can increase $N$ to improve accuracy, it will significantly enlarge the size of the system and introduce complexity. Thus, it is essential to achieve the right balance between computational complexity and dynamical precision. The effect of truncation order on solution convergence and sparsity of the matrix in shown in Fig. \ref{fig:threefigs}.
\begin{figure}[h!]
    \centering
    \begin{minipage}{1\textwidth}
        \centering
        \includegraphics[width=\linewidth]{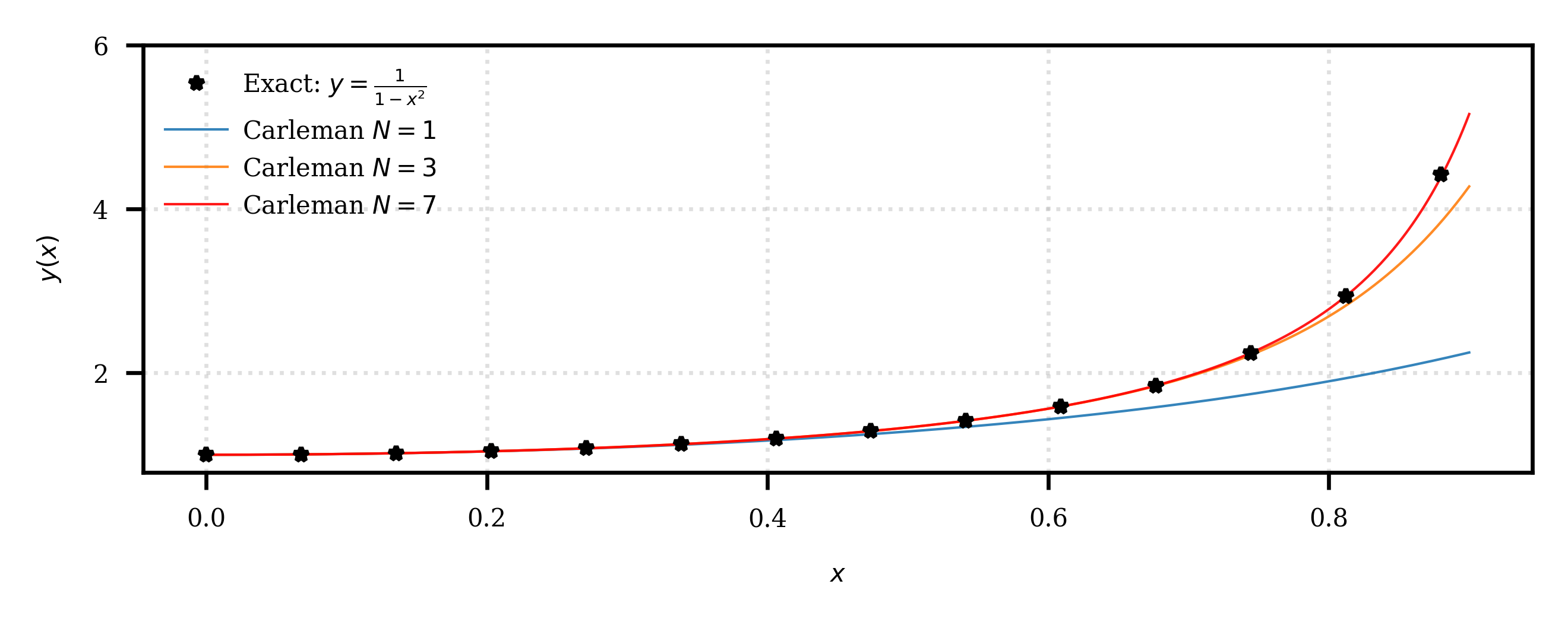}
    \end{minipage}
    \hspace{0.5em}
    \begin{minipage}{0.45\textwidth}
        \centering
        \includegraphics[width=\linewidth]{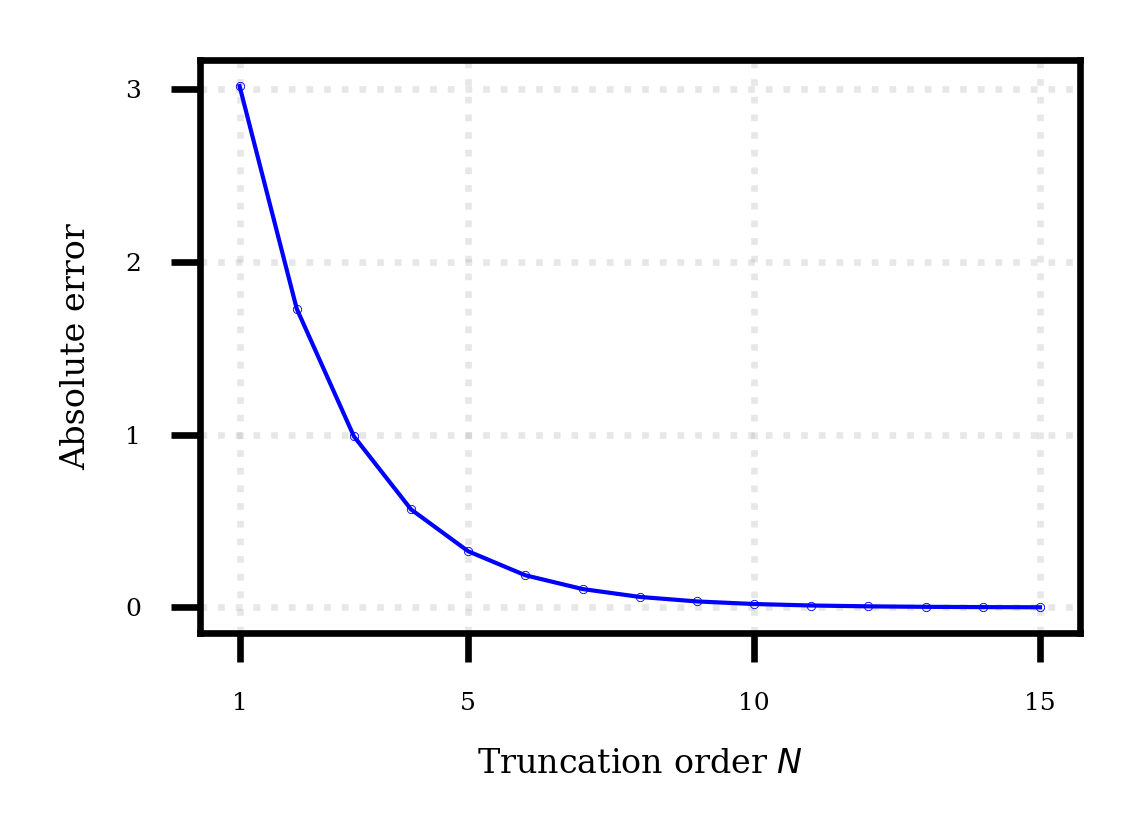}
    \end{minipage}
    \begin{minipage}{0.45\textwidth}
        \centering
        \includegraphics[width=\linewidth]{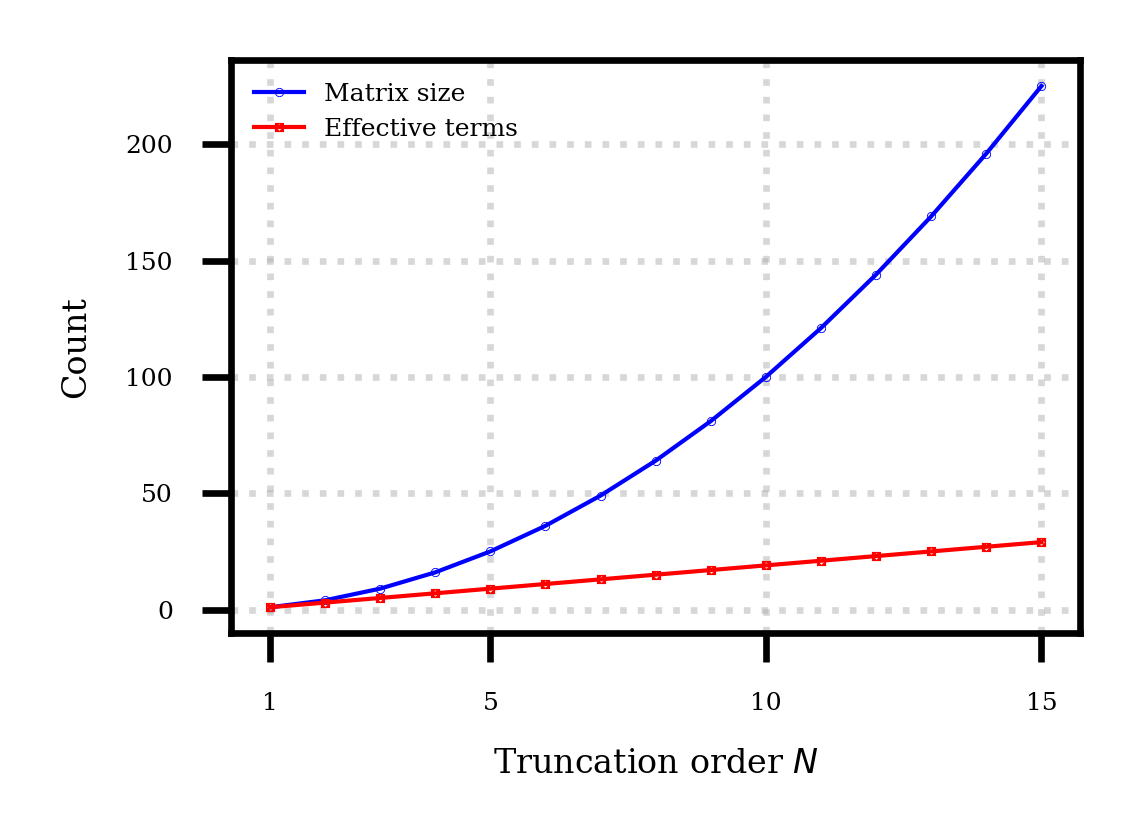}
    \end{minipage}

    \caption{Top: Systematic convergence of Carleman approximation for quadratic model of Bernoulli's equation with $P(x) = 2x$, $Q(x)=2x^{3}$, and $y(0)=1$. The solution is plotted for $x \geq 0$. Bottom Left: Progressive error reduction at increasing values of $N$. Bottom Right: At increasing order of $N$, matrix dimension scales quadratically, however, matrix remains sparse containing only $2N-1$ number of non-zero terms.(\textit{source code: \url{https://github.com/ali-tayyab/Carleman-Embedding}})}
    \label{fig:threefigs}
\end{figure}
\section{Towards Efficient Data Access Model}\label{section-3}
Although Eqn.(\ref{final-truncated-matrix}) reflects a linear system and there exist quantum algorithms to solve linear problems\cite{berry2014high,berry2017quantum,xin2020quantum,harrow2009quantum}, the solution of a Carleman-linearized system is still prohibited due to its quadratic growth over truncation order. As a result of Carleman embedding, we get a large sparse structured matrix which should be truncated at a higher value of $N$ to achieve better accuracy, making it infeasible to load this much data from a classical computer to a quantum computer. Hence, efficient data access models that can somehow reduce computational resources for data loading need to be discussed. This is the first post-linearization problem that one has to face if the linear system comes from the Carleman embedding routine. The underlying technique also does not guaranty that the resultant matrix after linearization is unitary. This is a common problem as most of the operators originating from physics or differential equations are non-unitary. However, quantum computers are capable of implementing only unitary operations. For this, we decompose our Hamiltonian and express it in terms of some unitary operators to process it on a quantum computer.

In this section, we discuss the first and foremost data access model in which the Hamiltonian is decomposed into the weighted sum of the \textit{Pauli basis}. It enables us to implement arbitrary operators on a quantum computer. However, in this case, the decomposition terms grow quadratically with the matrix size. Hence, we then provide a new basis set of non-unitary operators, which can better capture the structure of the Hamiltonian and provide efficient decomposition by minimizing the number of decomposition terms as compared to the Pauli basis.
\subsection{Linear Combination of Unitaries (LCU)}\label{section-3.1}
All we do in quantum computing is to simply process the quantum information using unitary operations. However, not all Hamiltonian models preserve norm and reversibility. LCU is an algorithmic primitive which allows us to represent a non-unitary operator in terms of some known unitaries to implement it on quantum states. This makes LCU a versatile tool that is used in the design of quantum algorithms. The main idea is the decomposition of a non-unitary Hamiltonian into a linear combination of unitary operators--probably the Pauli operators, as they form an orthonormal basis for all $2^{n} \times 2^{n}$ matrices for $n$ qubits, and we can decompose any matrix into these basis. Consider a non-unitary operator $\mathcal{H}$, then the decomposition will look like 
\begin{equation}\label{decomposition0}
    \mathcal{H} = \sum_{k=0}^{m-1} c_{k} U_{k}
\end{equation}
where $U$ satisfies $UU^{\dagger} =1$, and $c_{k} = \frac{1}{d} \textit{Tr}(U_{k}^{\dagger}\mathcal{H}) \in \mathbb{R}$ for $\mathcal{H}\mathcal{H}^{\dagger}=1$ where $\mathcal{H}$ acts on a d-dimensional Hilbert space. To proceed further, we use Block Encoding, another algorithmic primitive that encodes the decomposed Hamiltonian as a block of a large unitary operator, to apply this combination of unitaries with respective coefficients $c_{k}$. In this technique, we use a simple rule: Do, Apply, Undo. We introduce $\log_{2} m$ auxiliary qubits, with $m$ being the number of decomposition terms, and prepare them in a state whose amplitudes reflect the magnitude of the coefficients $c_{k}$. These qubits will help us to choose the right unitaries from our decomposition terms to apply with the right coefficients. This can be done by defining two operators, namely $PREP$ and $SELECT$. The $PREP$ will prepare the ancillary qubits in the state $\ket{\psi}$ by $PREP \ket{0}^{\otimes(\log_{2} m)}$, so that
\begin{equation}\label{ancilla-state}
    \ket{\psi} := \sum_{i=0}^{m-1} \sqrt{\alpha_{i}} \ket{i}
\end{equation}
where $\alpha_{i} = |c_{k}|/ \sum_{j} |c_{j}|$.
We then select the unitaries $U_{k}$ which we need to apply to our system qubits based upon the ancilla values.
\begin{equation} \label{SEL-op}
    SEL (\mathbb{P}) = \sum_{i=0}^{m-1} \ket{i} \bra{i} \otimes \mathbb{P}_{i}
\end{equation}
where $\mathbb{P}_{i} = \otimes_{i}\sigma_{i},\, \forall \sigma_{i} \in \{I, \sigma_{x}, \sigma_{y}, \sigma_{z}\}$, and $\sigma_{x}$, $\sigma_{y}$, and $\sigma_{z}$ denote the Pauli-X, Pauli-Y, and Pauli-Z operators, respectively\footnote{In rest of the paper, we'll denote the Pauli operators set by $\mathbb{P}$.}.
Lastly, we ``undo'' the ancilla superposition using $PREP^{\dagger}$. This will efficiently leave the effect of $\mathcal{H}$ on system qubits.

For the case in point, let us consider a $4 \times 4$ non-unitary sparse Hamiltonian and decompose it into the weighted sum of the Pauli basis.
\begin{equation}\label{decomposition}
    \mathcal{H} = \begin{pmatrix}
        1 &  &  & 0.5 \\
        & 0 &  &  \\
        &  & 0 &  \\
        0.5 &  &  & -1
    \end{pmatrix} = \frac{1}{2}\left(I \otimes \sigma_{z}\right) + \frac{1}{4}\left(\sigma_{x} \otimes \sigma_{x}\right) - \frac{1}{4}\left(\sigma_{y} \otimes \sigma_{y}\right) + \frac{1}{2}\left(\sigma_{z} \otimes I\right)
\end{equation}
Each component in the decomposition of $\mathcal{H}$ can be applied separately in the following way:
\begin{figure}[h!]
    \centering
    \begin{subfigure}
        \centering
        \scalebox{0.85}{
        \begin{quantikz}
            &  &
            \gategroup[2,steps=3,style={dashed,rounded
            corners,fill=blue!20, inner
            xsep=2pt},background,label style={label
            position=below,anchor=north,yshift=-0.2cm}]{{\sc $I \otimes \sigma_{z}$}} & \gate{I} & & & \\
            & &  & \gate{Z} & & &  
        \end{quantikz}
        }
    \end{subfigure}
    \begin{subfigure}
        \centering
        \scalebox{0.85}{
        \begin{quantikz}
            &  &
            \gategroup[2,steps=3,style={dashed,rounded
            corners,fill=blue!20, inner
            xsep=2pt},background,label style={label
            position=below,anchor=north,yshift=-0.2cm}]{{\sc $\sigma_{x} \otimes \sigma_{x}$}} & \gate{X} &  & & \\
            & &  & \gate{X} &  &  &
        \end{quantikz}
        }
    \end{subfigure}
    \begin{subfigure}
        \centering
        \scalebox{0.85}{
        \begin{quantikz}
            &  &
            \gategroup[2,steps=3,style={dashed,rounded
            corners,fill=blue!20, inner
            xsep=2pt},background,label style={label
            position=below,anchor=north,yshift=-0.2cm}]{{\sc $\sigma_{y} \otimes \sigma_{y}$}} & \gate{Y} &  & & \\
            & &  & \gate{Y} &  &  &
        \end{quantikz}
        }
    \end{subfigure}
    \begin{subfigure}
        \centering
        \scalebox{0.85}{
        \begin{quantikz}
            &  &
            \gategroup[2,steps=3,style={dashed,rounded
            corners,fill=blue!20, inner
            xsep=2pt},background,label style={label
            position=below,anchor=north,yshift=-0.2cm}]{{\sc $\sigma_{z} \otimes I$}} & \gate{Z} &  & & \\
            & &  & \gate{I} &  &  &
        \end{quantikz}
        }
    \end{subfigure}
\end{figure}

However, for their weighted sum, we prepare an auxiliary state as defined in Eqn.(\ref{ancilla-state}),
\begin{equation}\label{auxiliary-state-example}
\ket{\psi} = \sqrt{\frac{1}{2}}\ket{00} + \sqrt{\frac{1}{4}} \ket{01} + \sqrt{\frac{1}{4}} \ket{10} + \sqrt{\frac{1}{2}} \ket{11}
\end{equation}
To see how this logic will work, we consider an arbitrary state $\ket{\chi}$ and compute
\begin{equation}\label{}
    \textbf{U} = \left(\bra{0} \otimes \mathbb{I}\right) PREP^{\dagger}. SEL. PREP \left(\ket{0} \otimes \ket{\chi}\right)
\end{equation}
where $PREP$ acts on ancilla qubits only. This operator will prepare our ancilla qubits in the state $\ket{\psi}$, such that\footnote{$PREP\ket{0} = \ket{\psi}$ and $\bra{0}PREP^{\dagger} = \bra{\psi}$.}
$$
    = \left[\sqrt{\frac{1}{2}}\bra{00} + \sqrt{\frac{1}{4}} \bra{01} + \sqrt{\frac{1}{4}} \bra{10} + \sqrt{\frac{1}{2}} \bra{11}\right]SEL\left[\left(\sqrt{\frac{1}{2}}\ket{00} + \sqrt{\frac{1}{4}} \ket{01} + \sqrt{\frac{1}{4}} \ket{10} + \sqrt{\frac{1}{2}} \ket{11}\right) \otimes \ket{\chi}\right]
$$
Using Eqn.(\ref{SEL-op}), $SEL(\mathbb{P})$ will become $(\ket{0}\bra{0}\mathbb{P}_{0} + \ket{0}\bra{1}\mathbb{P}_{1} + \ket{1}\bra{0}\mathbb{P}_{2} + \ket{1}\bra{1}\mathbb{P}_{3})$, and hence

$$
= \left[\sqrt{\frac{1}{2}}\bra{00} + \sqrt{\frac{1}{4}} \bra{01} + \sqrt{\frac{1}{4}} \bra{10} + \sqrt{\frac{1}{2}} \bra{11}\right] \left[\left(\sqrt{\frac{1}{2}} \ket{00} \mathbb{P}_{0} + \sqrt{\frac{1}{4}} \ket{01} \mathbb{P}_{1}  + \sqrt{\frac{1}{4}} \ket{10} \mathbb{P}_{2}  + \sqrt{\frac{1}{2}} \ket{11} \mathbb{P}_{3}\right) \otimes \ket{\chi} \right]
$$
\begin{equation}\label{sum_of_unitaries}
    = \left(\frac{1}{2} \mathbb{P}_{0} + \frac{1}{4} \mathbb{P}_{1} + \frac{1}{4} \mathbb{P}_{2} + \frac{1}{2} \mathbb{P}_{3}\right) \ket{\chi}
\end{equation}
\begin{equation}\label{final-eq0}
     = \mathcal{H} \ket{\chi}
\end{equation}
which is precisely Eqn.(\ref{decomposition}) with $\{\mathbb{P}_{0},\mathbb{P}_{1},\mathbb{P}_{2},\mathbb{P}_{3}\} = \{I \otimes \sigma_{z}, \sigma_{x} \otimes \sigma_{x}, \sigma_{y} \otimes \sigma_{y}, \sigma_{z} \otimes I\}$. If we consider $\mathbb{I}$ as an arbitrary state instead of $\ket{\chi}$,
\begin{equation}\label{final-eq1}
    (\bra{0} \otimes \mathbb{I}) PREP^{\dagger}.SEL.PREP(\ket{0} \otimes \mathbb{I}) = \mathcal{H}
\end{equation}
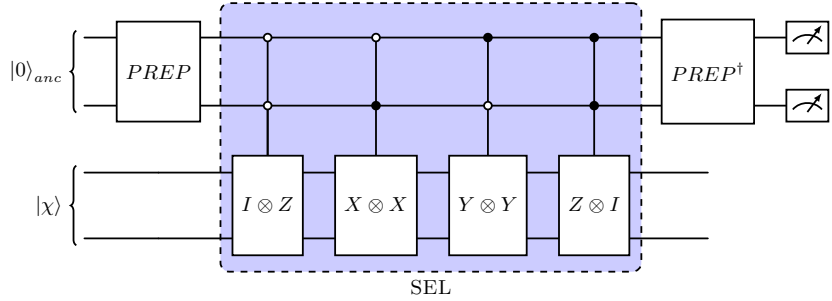
\begin{figure}[h!]
\centering
\scalebox{0.85}{
\begin{quantikz}
\lstick[2]{$\ket{0}_{anc}$}& \gate[2]{PREP} & \ctrl[open]{3}\gategroup[4,steps=4,style={dashed,rounded
corners,fill=blue!20, inner
xsep=2pt},background,label style={label
position=below,anchor=north,yshift=-0.2cm}]{{\sc
SEL}} & \ctrl[open]{1} & \ctrl{1} & \ctrl{1} & \gate[2]{PREP^{\dagger}} & \meter{} \\
&  & \ctrl[open]{1} & \ctrl{1} & \ctrl[open]{1} & \ctrl{1} & & \meter{}\\
\lstick[2]{$\ket{\chi}$}& & \gate[2]{I \otimes Z}& \gate[2]{X \otimes X} & \gate[2]{Y \otimes Y} & \gate[2]{Z \otimes I} & \\
& &  & & & &
\end{quantikz}
}
\caption{Circuit for controlled-$\mathbb{P}_{i}$ to implement the Hamiltonian described in Eqn.(\ref{decomposition}).}
\end{figure}
The joint operator $\textbf{U}$ acting on the total space, i.e, $\mathcal{H}_{total} = \mathcal{H}_{anc} \otimes \mathcal{H}_{sys}$, followed by the n-dimensional ancilla space and the N-dimensional system space, can be described as a block matrix determined by the ancilla basis. The action of $\textbf{U}$ on any arbitrary separable state, say $(\ket{j} \otimes \ket{\psi})$ will describe the evolution of the system under the ancilla transition $\ket{j} \xrightarrow{} \ket{i}$ using a block $U_{ji}$ lying inside the big $\textbf{U}$. In our case, Eqn.(\ref{final-eq1}) represents the projection of an ancilla back onto $\ket{0}$, giving the affective action $U_{00}$ as $\mathcal{H}$ that lies in the upper left corner of the large unitary $\textbf{U}$.
\begin{equation}
    \textbf{U} = \begin{bmatrix}
        \mathcal{H} & * \\
        * & *
    \end{bmatrix}
\end{equation}
The rest of the blocks of $\textbf{U}$ make it unitary, and since the ancilla starts and ends with $\ket{0}$, these blocks did not affect the encoded subspace.

Returning to our original problem, we need to implement a non-unitary coming from a large truncated system with $(2N+1)$ degrees of freedom/effective terms. Apart from how efficiently Block Encoding handles the implementation of the weighted tensor product of unitaries, the LCU decomposition in $Pauli\,basis$ usually results in quadratic growth of decomposition terms, which we cannot bear at this stage, as we already have a large system to work with. Let us illustrate this with a couple of examples. Consider the following sparse matrices,
\begin{equation*}
A_{1} =
\begin{pmatrix}
 2 &   &   &   \\
   & 0 &   &   \\
   &   & 0 &   \\
   &   &   & 7
\end{pmatrix}
,\qquad
A_{2} =
\begin{pmatrix}
 1 &   &   &   \\
 4 &   &   &   \\
 2 &   &   &   \\
 1 & -2 & 1 & -1
\end{pmatrix}
,\qquad
A_{3} =
\begin{pmatrix}
 1 &   &   &   \\
 4 & 3 &   &   \\
 2 & -2 & 2 &   \\
 1 & 1 & -1 & 1
\end{pmatrix}
\end{equation*}
The corresponding decomposition into the linear combination of $\mathbb{P}$ is given below.
\begin{equation}\label{eq:examples}
\begin{aligned}
A_{1}
&= 2.25 (I \otimes I)
 - 1.25 (I \otimes \sigma_{z})
 - 1.25 (\sigma_{z} \otimes I)
 + 2.25 (\sigma_{z} \otimes \sigma_{z}) \\
A_{2}
&= 1.25 (I \otimes \sigma_{x})
 - 1.25\, i (I \otimes \sigma_{y})
 + 0.5 (I \otimes \sigma_{z})
 + 0.25 (\sigma_{x} \otimes \sigma_{x})
 - 0.25 \, i (\sigma_{x} \otimes \sigma_{y})
+ (\sigma_{x} \otimes \sigma_{z})\\
&\quad
 - 0.25 \, i (\sigma_{y} \otimes \sigma_{x}) 
 - 0.25 (\sigma_{y} \otimes \sigma_{y})
 - i (\sigma_{y} \otimes \sigma_{z})
 + 0.5 (\sigma_{z} \otimes I) 
+ 0.75 (\sigma_{z} \otimes \sigma_{x})
 - 0.75 (\sigma_{z} \otimes \sigma_{y})\\
A_{3}
&= 1.75 (I \otimes I)
 + 0.75 (I \otimes \sigma_{x})
 - 0.75 \, i(I \otimes \sigma_{y})
 - 0.25 (I \otimes \sigma_{z}) 
+ 0.75 (\sigma_{x} \otimes I)
 - 0.25 (\sigma_{x} \otimes \sigma_{x})\\
&\quad
 - 0.75 \, i(\sigma_{x} \otimes \sigma_{y})
 + 0.25 (\sigma_{x} \otimes \sigma_{z}) 
- 0.75 \,i(\sigma_{y} \otimes I)
 + 0.25 \,i(\sigma_{y} \otimes \sigma_{x})
 - 0.75 (\sigma_{y} \otimes \sigma_{y})
 - 0.25 \,i(\sigma_{y} \otimes \sigma_{z})\\
&\quad
 + 0.25 (\sigma_{z} \otimes I)
 + 1.25 (\sigma_{z} \otimes \sigma_{x})
 - 1.25 \,i(\sigma_{z} \otimes \sigma_{y})
 - 0.75 (\sigma_{z} \otimes \sigma_{z})
\end{aligned}
\end{equation}

For a sparse structured matrix, an ideal decomposition would take the same number of terms as the non-zero entries in the matrix. However, Eqn.(\ref{eq:examples}) illustrates that this is not the case when we use the LCU decomposition in the tensor product of the Pauli basis. As a drawback, we need more ancillae and gates to implement during the circuit construction for Block Encoding, eventually making the circuit deep by increasing resource requirements. One could also face large statistical noise and slower convergence while implementing the circuit containing a large number of unitaries.
\subsection{Linear Combination of Non-Unitaries (LCNU)}\label{section-3.2}
We define a set of \textit{non-unitary} operators, namely the Sigma basis, which could result in ideal decomposition.
\begin{dfn}\label{definition-01}
    The Sigma basis set $\mathbb{S}$ is the set comprising the non-unitary operators\cite{liu2021variational},
\begin{equation*}
\mathbb{S} = \{\mathbb{I}_{2}, \sigma_{+},\sigma_{-}, \sigma_{+}\sigma_{-}, \sigma_{-}\sigma_{+}\}     
\end{equation*}
where $\mathbb{I}_{2}$ is the standard single qubit identity, and the rest are outer product operators defined as  
\begin{equation*}
\sigma_{+} = \ket{0}\bra{1} = 
    \begin{pmatrix}
        0 & 1 \\
        0 & 0
    \end{pmatrix}
    \quad
    \sigma_{-} = \ket{1}\bra{0} =
    \begin{pmatrix}
        0 & 0\\
        1 & 0
    \end{pmatrix}
\end{equation*}

\begin{equation*}
\sigma_{+}\sigma_{-} = \ket{0}\bra{0} = 
    \begin{pmatrix}
        1 & 0 \\
        0 & 0
    \end{pmatrix}
    \quad
    \sigma_{-}\sigma_{+} = \ket{1}\bra{1} =
    \begin{pmatrix}
        0 & 0\\
        0 & 1
    \end{pmatrix}
\end{equation*}
\end{dfn}
These operators, in contrast to $\mathbb{P}$, are localized to single entries, allowing us to assign a single operator to the single entry of the matrix. In this way, the decomposition will take as many terms as the matrix has non-zero entries, showing the linear growth. Hence, in the tensor product of $\mathbb{S}$, the analog of Eqn.(\ref{decomposition0}) will be
\begin{equation}\label{LCNU}
    \mathcal{H} = \sum_{j = 1}^{J} \alpha_{j} \mathcal{H}_{j}, \,\,\,\, \mathcal{H}_{j} = \otimes_{j}\sigma_{j},
\end{equation}
where $\sigma_{j} \in \mathbb{S}$. To illustrate $ J \leq |supp(\mathcal{H})|$, where $J$ is the total number of decomposition terms while the cardinality of $supp(\mathcal{H})$ represents the non-zero terms in $\mathcal{H}$, we again decompose the matrices described in Eqn.(\ref{eq:examples}) into the linear combination of Sigma basis\footnote{using $\sigma_{\pm} = \sigma_{+} \sigma_{-}$, and $\sigma_{\mp} = \sigma_{-}\sigma_{+}$ for convenience.}.
\begin{equation}
\begin{aligned}
A_{1} &= 2 \left(\sigma_{\pm} \otimes \sigma_{\pm}\right) + 7 \left(\sigma_{\mp} \otimes \sigma_{\mp}\right) \\
A_{2} &= (\sigma_{-} \otimes \sigma_{+}) + 2(\sigma_{-} \otimes \sigma_{\pm}) - 2(\sigma_{-} \otimes \sigma_{\mp}) + 4(\sigma_{\pm} \otimes \sigma_{-}) + (\sigma_{\pm} \otimes \sigma_{\pm}) \\
&\quad + (\sigma_{\mp} \otimes \sigma_{-}) - (\sigma_{\mp} \otimes \sigma_{\mp}) \\
A_{3} &= -2 (\sigma_{-} \otimes \sigma_{+}) + (\sigma_{-} \otimes \sigma_{-}) + 2 (\sigma_{-} \otimes \sigma_{\pm}) + (\sigma_{-} \otimes \sigma_{\mp}) + 4 (\sigma_{\pm} \otimes \sigma_{-}) \\
&\quad + (\sigma_{\pm} \otimes \sigma_{\pm}) + 3 (\sigma_{\pm} \otimes \sigma_{\mp}) - (\sigma_{\mp} \otimes \sigma_{-}) + 2 (\sigma_{\mp} \otimes \sigma_{\pm}) + (\sigma_{\mp} \otimes \sigma_{\mp})
\end{aligned}
\end{equation}
In each case, the decomposition cost scales linearly with non-zero entries of the matrix, providing an efficient decomposition compared to the Pauli basis in which decomposition terms grow quadratically or worse, making the whole system resource-inefficient. This shows that a sparse matrix having $p$ non-zero terms will have exactly $p$ terms in the Sigma basis expansion, showing an advantage over Pauli basis decomposition.

\begin{figure}[h!]
    \centering
    \includegraphics[width=0.45\linewidth]{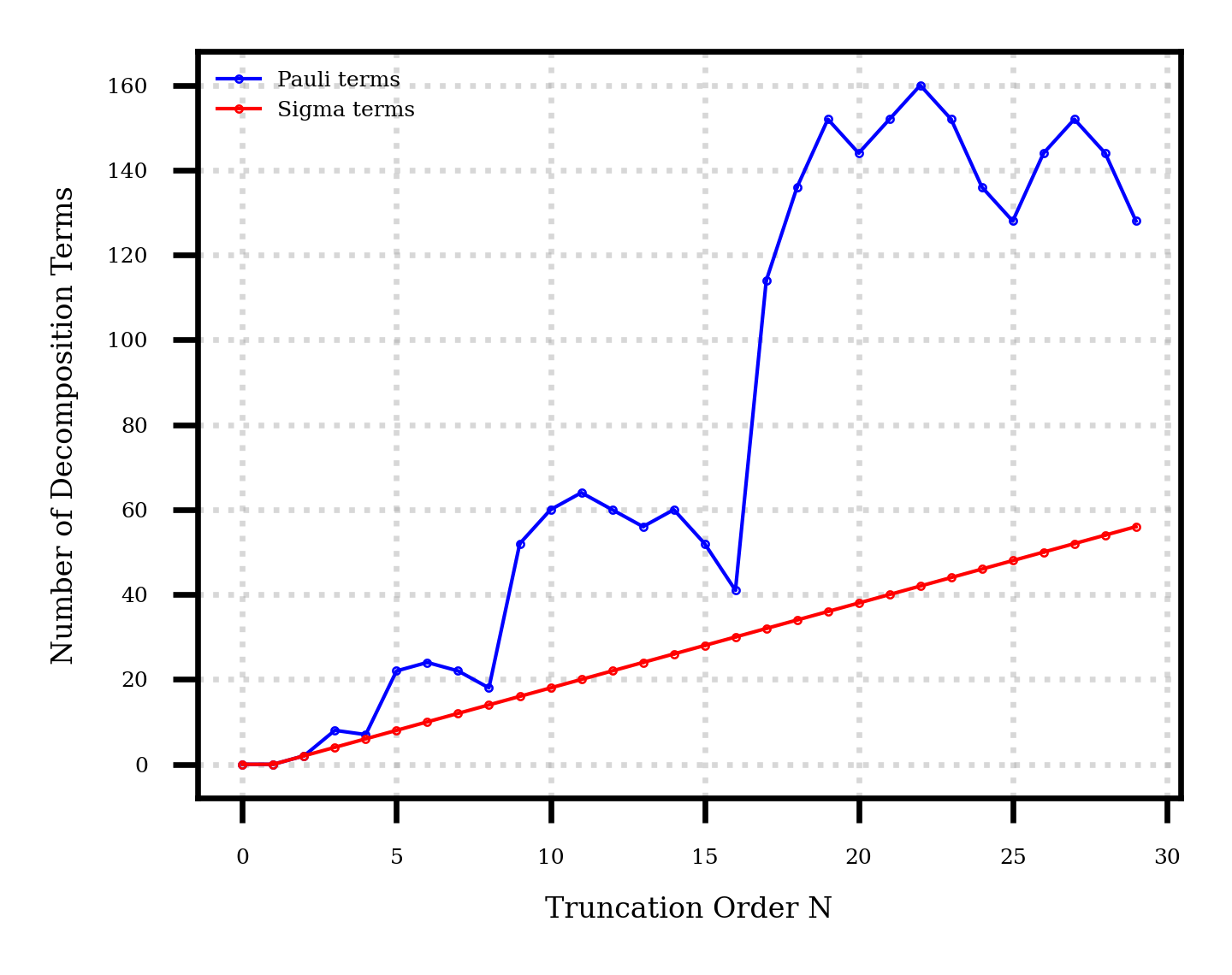}
    \caption{Difference between the number of terms in Sigma and Pauli basis decomposition for the quadratic model of Bernoulli's equation at increasing truncation order.(\textit{source code: \url{https://github.com/ali-tayyab/Carleman-Embedding}})}
    \label{fig:placeholder}
\end{figure}

However, $\mathbb{S}$ contains projection operators $\ket{j}\bra{k}$, which do not preserve the quantum state norm, i.e., they are \textit{not} unitary (except $\mathbb{I}_{2}$), and hence cannot be implemented as a gate by itself. Decomposing these non-unitaries back into \textit{Pauli basis} will destroy all the advantages as each component of $\mathbb{S}$ requires multiple Pauli operators for its decomposition, which will eventually take us back to our original problem of inefficient decomposition or quadratic growth of decomposition terms over matrix size. Hence, we must look at some other techniques to implement the tensor product of these non-unitaries in order to harness the efficiency of Sigma-decomposition. One such technique is unitary completion--an approach mathematically equivalent to block encoding which can manipulate the non-unitary operators to make them implementable using quantum circuits. In this approach, we can directly modify the non-unitary operator by finding the rows and columns that could be added to the matrix to make it unitary. Unlike block encoding, it works in the same Hilbert space and does not require any ancilla qubits, which makes it convenient and efficient for small matrices.
\subsubsection{Unitary Completion and Circuit Construction for $\mathcal{H}_{j}$}\label{section-3.2.1}
\begin{dfn}\label{definition-02}
    Consider a linear operator $\textbf{O}:U \xrightarrow{} V$, where $U$ and $V$ are complex vector spaces within $\mathbb{C}^{n}$ such that $U \subset V \subseteq \mathbb{C}^{n}$, which satisfy 
    $$
    \bra{u_{1}} \textbf{O}^{\dagger}\textbf{O} \ket{u_{2}} = \langle u_{1} | u_{2}\rangle\,\,\, \forall \ket{u_{1}},\ket{u_{2}} \in U
    $$
    then there always exists a unitary operator $\Bar{\textbf{O}}$ that serves as a unitary completion for \textbf{O}, whose action on $\ket{u}$ is the same as the action of \textbf{O} on $\ket{u}$.
    $$
    \textbf{O} \ket{u} = \Bar{\textbf{O}}\ket{u}
    $$
\end{dfn}
Using this definition, we can now define the resultant unitary matrix corresponding to $\mathcal{H}_{j}$ after unitary completion;
\begin{equation}
\bold{U}_{j} =
    \begin{bmatrix}
         \mathcal{H}_{j} & \Bar{\mathcal{H}_{j}} - \mathcal{H}_{j} \\
         \Bar{\mathcal{H}_{j}} - \mathcal{H}_{j} & \mathcal{H}_{j}
    \end{bmatrix}
\end{equation}
where $\Bar{\mathcal{H}_{j}}$ represents the unitary completion for $\mathcal{H}_{j}$. Since $\mathcal{H}_{j} = \otimes_{p} \sigma_{p}$, where $\sigma_{p} \in \mathbb{S}$, we need to define $\Bar{\sigma_{p}}$ for each component of $\mathbb{S}$.
\begin{theorem}\label{theorem-1}
The unitary completion for $\mathcal{H}_{j} = \otimes_{p}\sigma_{p}$, $\sigma_{p} \in \mathbb{S}$ is given by $\Bar{\mathcal{H}_{j}} = \otimes_{p}\Bar{\sigma_{p}}$, such that $\Bar{\sigma_{p}} = \sigma_{x}$ for $\{\sigma_{+}, \sigma_{-}\}$, and $\Bar{\sigma_{p}} = I$ for the rest of the $\mathbb{S}$ operators, where $I$ is the standard identity while $\sigma_{x}$ is the Pauli-X operator of $\mathbb{P}$. In this way, we define a new set $\Bar{\mathbb{S}}$, which comprises the unitary completions for $\mathbb{S}$.
$$
\Bar{\mathbb{S}} = \{I, \sigma_{x}, \sigma_{x}, I, I\}
$$
\end{theorem}
\begin{proof}
    Corresponding to \textbf{O} in Definition (\ref{definition-02}), there exists an orthogonal complement $\textbf{O}^{\;'}:= \Bar{\textbf{O}} - \textbf{O}$, and hence for each component of $\mathbb{S}$, such that
    $$
    \Bar{\sigma_{p}}=\sigma_{p} + \sigma_{p}^{\;'},\,\, \sigma_{p} \in \mathbb{S}
    $$
The claim is as follows.
        $$
        \sigma_{p}^{\;'} = \begin{cases}
            \sigma_{+} \leftrightarrow \sigma_{-}, \\ 
            \sigma_{+}\sigma_{-} \leftrightarrow \sigma_{-}\sigma_{+}
        \end{cases}
        $$
To justify it, we compute the unitary completion for each component of $\mathbb{S}$.
\begin{equation}
    \begin{split}
        \overline{\sigma_{+}\sigma_{-}} = \sigma_{+}\sigma_{-} + \sigma_{-}\sigma_{+} = \begin{pmatrix}
                1 & 0 \\
                0 & 1
            \end{pmatrix} = I \\
            \overline{\sigma_{-}\sigma_{+}} = \sigma_{-}\sigma_{+} + \sigma_{+}\sigma_{-} = \begin{pmatrix}
                1 & 0 \\
                0 & 1
            \end{pmatrix} = I \\
            \Bar{\sigma_{+}} = \sigma_{+} + \sigma_{-} = \begin{pmatrix}
                0 & 1 \\
                1 & 0
            \end{pmatrix} = \sigma_{x} \\
            \Bar{\sigma_{-}} = \sigma_{-} + \sigma_{+} = \begin{pmatrix}
                0 & 1 \\
                1 & 0
            \end{pmatrix} = \sigma_{x}
    \end{split}
\end{equation}
 Since $\mathbb{I}_{2}$ is the unitary operator of $\mathbb{S}$, its orthogonal complement will be \textbf{0}. In this way, we can define an ordered set $\mathbb{S}^{\;'}$ that contains the orthogonal complement of each component of $\mathbb{S}$.
   \begin{equation}
       \mathbb{S}^{\;'} = \{\bold{0}, \sigma_{-},\sigma_{+},\sigma_{-}\sigma_{+},\sigma_{+}\sigma{-}\}
   \end{equation}
\end{proof}
The action of $U_{j}$ on any arbitrary state, say $\ket{0} \otimes \ket{\chi}$, is given by
\begin{equation}\label{action-U}
    U_{j} \left(\ket{0} \otimes \ket{\chi}\right) = \ket{0}\mathcal{H}_{j}\ket{\chi} +  \ket{1} \Bar{\mathcal{H}_{j}} - \mathcal{H}_{j} \ket{\chi}
\end{equation}
To construct the circuit for $U_{j}$ efficiently, we split this unitary into the product of two unitaries.
\begin{equation}\label{split-unitary}
U_{j} = U_{j,a} U_{j,b}
\end{equation}
where,
\begin{equation}\label{1-splitter}
        U_{j,b} = \sigma_{x} \otimes \Bar{\mathcal{H}_{j}} = \begin{pmatrix}
            & \mathcal{H}_{j}^{'} + \mathcal{H}_{j} \\
            \mathcal{H}_{j}^{'} + \mathcal{H}_{j} & 
        \end{pmatrix}
\end{equation}
$$
    U_{j,a} = U_{j} (U_{j,b})^{T} = \begin{pmatrix}
         \mathcal{H}_{j} & \mathcal{H}_{j}^{'} \\
         \mathcal{H}_{j}^{'} & \mathcal{H}_{j}
    \end{pmatrix} \begin{pmatrix}
            & \mathcal({H}_{j}^{'})^{T} + \mathcal{H}_{j}^{T} \\
            (\mathcal{H}_{j}^{'})^{T} + \mathcal{H}_{j}^{T} & 
        \end{pmatrix}
$$
$$
 = \begin{pmatrix}
     \mathcal{H}_{j}^{'} (\mathcal{H}_{j}^{'})^{T} + \mathcal{H}_{j}^{'}\mathcal{H}_{j}^{T} & \mathcal{H}_{j}(\mathcal{H}_{j}^{'})^{T} + \mathcal{H}_{j} \mathcal{H}_{j}^{T} \\[3pt]
     \mathcal{H}_{j}(\mathcal{H}_{j}^{'})^{T} + \mathcal{H}_{j}\mathcal{H}_{j}^{T} & \mathcal{H}_{j}^{'}(\mathcal{H}_{j}^{'})^{T} + \mathcal{H}_{j}^{'} \mathcal{H}_{j}^{T} 
 \end{pmatrix}
$$
Here, $\mathcal{H}_{j}^{'} \mathcal{H}_{j}^{T} = \mathcal{H}_{j}(\mathcal{H}_{j}^{'})^{T} = 0$ (see Lemma. 1,\cite{gnanasekaran2024efficient}). Hence,
$$
U_{j,a}=\begin{pmatrix}
     \mathcal{H}_{j}^{'} (\mathcal{H}_{j}^{'})^{T} &  \mathcal{H}_{j} \mathcal{H}_{j}^{T} \\[3pt]
      \mathcal{H}_{j}\mathcal{H}_{j}^{T} & \mathcal{H}_{j}^{'}(\mathcal{H}_{j}^{'})^{T}  
 \end{pmatrix}
$$
\begin{theorem}\label{orthogonal-completeness}
    For $\sigma_{p} \in \mathbb{S}$ and $\sigma_{p}^{'} \in \mathbb{S}^{'}$, $(\mathcal{H}_{j}, \mathcal{H}_{j}^{'})$ forms a complete orthogonal pair, satisfying
    $$
    \mathcal{H}_{j}^{'} (\mathcal{H}_{j}^{'})^{T} + \mathcal{H}_{j} \mathcal{H}_{j}^{T} = I
    $$
\end{theorem}
\begin{proof}
    For $\sigma_{+}$: 
    $$
    \sigma_{+}^{'} (\sigma_{+}^{'})^{T} + \sigma_{+} \sigma_{+}^{T} = \begin{pmatrix}
        0 & 0 \\
        1 & 0
    \end{pmatrix}\begin{pmatrix}
        0 & 1 \\
        0 & 0
    \end{pmatrix} + \begin{pmatrix}
        0 & 1 \\
        0 & 0
    \end{pmatrix} \begin{pmatrix}
        0 & 0 \\
        1 & 0
    \end{pmatrix} = \begin{pmatrix}
        1 & 0 \\
        0 & 1
    \end{pmatrix}
    $$
    The rest can also be proved in a similar way.
\end{proof}
According to Theorem. \ref{orthogonal-completeness}, $U_{j,a}$ will now take the form
\begin{equation}\label{Uja}
    U_{j,a} = \begin{pmatrix}
        I - \mathcal{H}_{j}\mathcal{H}_{j}^{T} & \mathcal{H}_{j}\mathcal{H}_{j}^{T} \\
         \mathcal{H}_{j}\mathcal{H}_{j}^{T}
         & I - \mathcal{H}_{j}\mathcal{H}_{j}^{T}
    \end{pmatrix}
\end{equation}
The circuit construction for $U_{j,b}$ is quite straightforward. It requires $n$ single qubit gates for $\Bar{\mathcal{H}_{j}}$, as $\Bar{\mathcal{H}_{j}} \in \Bar{\mathbb{S}}$, and one ancilla qubit for $\sigma_{x}$(see Eqn.\ref{1-splitter}). However, for $U_{j,a}$ we need to understand the structure of the matrix from Eqn.(\ref{Uja}). For every $\sigma_{p} \in \mathbb{S}$ in $\otimes_{p}\sigma_{p}$, $\mathcal{H}_{j}$ satisfies $\mathcal{H}_{j}\mathcal{H}_{j}^{T} \in \{\sigma_{+}\sigma_{-}, \sigma_{-}\sigma_{+}\}$. That is, for the $n+1$ qubit system, $U_{j,a}$ is a $2^{n + 1} \times 2^{n + 1}$ permutation matrix, which can be implemented using Toffoli gates only. Let us take a couple of examples and construct the circuit for $U_{j,a}$ and $U_{j,b}$ for each to form a full unitary $U_{j}$.
\begin{example}
We start from a $(3+1)$ qubit system. Consider $H_{j} = \sigma_{-} \otimes \sigma_{+}\sigma_{-} \otimes I$, and the corresponding unitary completion $\Bar{\mathcal{H}_{j}} = \sigma_{x} \otimes I \otimes I$ (see Theorem. \ref{theorem-1}). According to Eqn.(\ref{1-splitter}), $U_{j,b}$ will then take the form
$$
U_{j,b} = \sigma_{x} \otimes \Bar{\mathcal{H}_{j}} = \sigma_{x} \otimes \sigma_{x} \otimes I \otimes I
$$
which can be implemented using $2$ Pauli-X gates (one on the ancilla qubit and the other on one of the system qubit). On the other hand, for $U_{j,a}$, let us compute $\mathcal{H}_{j} \mathcal{H}_{j}^{T}$, 
$$
\mathcal{H}_{j} \mathcal{H}_{j}^{T} = \sigma_{-} \sigma_{-}^{T} \otimes \sigma_{+}\sigma_{-} (\sigma_{+}\sigma_{-})^{T} \otimes I
$$
$$
 = \begin{pmatrix}
     0 & 0 \\
     0 & 1
 \end{pmatrix} \otimes \begin{pmatrix}
     1 & 0 \\
     0 & 0
 \end{pmatrix} \otimes \begin{pmatrix}
     1 & 0 \\
     0 & 1
 \end{pmatrix}
$$
that yields a $8 \times 8$ matrix with only two non zero rows,
$$
\mathcal{H}_{j}\mathcal{H}_{j}^{T}= \begin{pmatrix}
\phantom{0} \\[-2mm]  
\phantom{0} \\[-2mm]
\phantom{0} \\[-2mm]
\phantom{0} \\[-2mm]
0 & 0 & 0 & 0 & 1 & 0 & 0 & 0 \\  
0 & 0 & 0 & 0 & 0 & 1 & 0 & 0 \\  
\phantom{0} \\[-2mm]
\phantom{0}
\end{pmatrix}
$$
These rows $R_{4}$ and $R_{5}$ correspond to the binary representation \{0100\} and \{0101\}, respectively. Here, the leftmost is the ancilla bit and the order is taken as $\ket{q_{anc} q_{1} q_{2} q_{3}}$. This means that the action of $\mathcal{H}_{j}$ will be non-zero on these basis states only. These active rows can be implemented using two $ C^{3} X$ gates, that is, flipping the ancilla qubit when the system qubits match the above binary pattern.
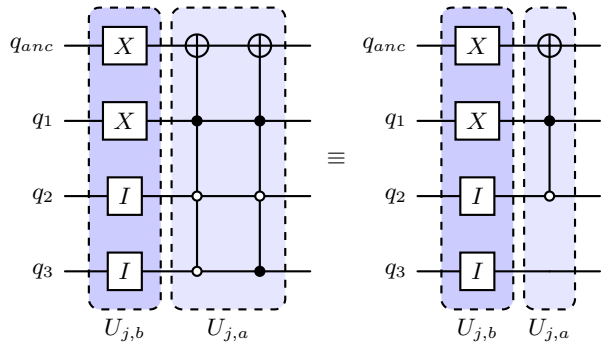
\begin{figure}[h!]
    \centering
    \begin{quantikz}
        \lstick{$q_{anc}$} & \gate{X}\gategroup[4,steps=1,style={dashed,rounded
         corners,fill=blue!20, inner
         xsep=2pt},background,label style={label
         position=below,anchor=north,yshift=-0.2cm}]{{\sc
         $U_{j,b}$}}  & \targ{} \gategroup[4,steps=2,style={dashed,rounded
         corners,fill=blue!10, inner
         xsep=2pt},background,label style={label
         position=below,anchor=north,yshift=-0.2cm}]{{\sc
         $U_{j,a}$}} & \targ{} & \\
        \lstick{$q_{1}$} & \gate{X}  & \ctrl{-1} & \ctrl{-1} &\\
        \lstick{$q_{2}$} & \gate{I} & \ctrl[open]{-1} & \ctrl[open]{-1} & \\
        \lstick{$q_{3}$} & \gate{I} & \ctrl[open]{-1} & \ctrl{-1} &
    \end{quantikz} $\equiv$  \begin{quantikz}
        \lstick{$q_{anc}$} & \gate{X}\gategroup[4,steps=1,style={dashed,rounded
         corners,fill=blue!20, inner
         xsep=2pt},background,label style={label
         position=below,anchor=north,yshift=-0.2cm}]{{\sc
         $U_{j,b}$}}  & \targ{} \gategroup[4,steps=1,style={dashed,rounded
         corners,fill=blue!10, inner
         xsep=2pt},background,label style={label
         position=below,anchor=north,yshift=-0.2cm}]{{\sc
         $U_{j,a}$}} & \\
        \lstick{$q_{1}$} & \gate{X}  & \ctrl{-1} & \\
        \lstick{$q_{2}$} & \gate{I} & \ctrl[open]{-1} & \\
        \lstick{$q_{3}$} & \gate{I} &  & 
    \end{quantikz}
    \caption{Circuit construction for $U_{j}$ corresponding to $\mathcal{H}_{j} = \sigma_{-} \otimes \sigma_{+}\sigma_{-} \otimes I$}
    \label{example-1-circuit-construction}
\end{figure}

\begin{theorem}\label{theorem-3}
    Two $C^{n} X$ which differ by one bit can combine to give one $C^{n-1} X$.
\end{theorem}
\begin{proof}
Let $\{b_{1}^{(1)}, b_{2}^{(1)}, \cdots, b_{n}^{(1)}\} \in \{0,1\}^{n}$, and $\{b_{1}^{(2)}, b_{2}^{(2)}, \cdots, b_{n}^{(2)}\} \in \{0,1\}^{n}$ be a binary representation of the two non-zero rows of the matrix $\mathcal{H}_{j}\mathcal{H}_{j}^{T}$, such that $b_{i}^{(1)} = b_{i}^{(2)}\, \forall i \neq n-1$, that is, both differ only in one bit. Let $q_{1},q_{2},\cdots,q_{n}$ be the system qubits, and $q_{anc}$ be the ancilla qubit that we want to flip according to conditions $q_{i} = b_{i}^{(1)}$ or $q_{i} = b_{i}^{(2)}$. We apply two $C^{n} X$ gates, such that the first is
$$
\oplus \,\,\, if (q_{1} = b_{1}^{(1)} \wedge q_{2} = b_{2}^{(1)} \wedge \cdots) 
$$
while the second one,
$$
\oplus\,\,\, if (q_{1} = b_{1}^{(2)} \wedge q_{2} = b_{2}^{(2)} \wedge \cdots)
$$
We can combine both of these conditions using an OR operation, such that
$$
(q_{1} = b_{1}^{(1)} \wedge q_{2} = b_{2}^{(1)} \wedge \cdots) \vee (q_{1} = b_{1}^{(2)} \wedge q_{2} = b_{2}^{(2)} \wedge \cdots)
$$
Since both differ in one bit, the rest can be taken as common as they satisfy $q_{i} = b_{i}^{(1)} = b_{i}^{(2)} \forall i \neq n-1$\footnote{Since $b_{i}^{(1)}= b_{i}^{(2)}$ except for $i = n-1$, we use $b_{i}^{(1)} = b_{i}^{(2)} = b_{i}$ $\forall i \neq n-1$.}.
$$
(q_{1} = b_{1} \wedge q_{2} = b_{2} \wedge \cdots ) \wedge (q_{n-1} = b_{n-1}^{(1)} \vee q_{n-1} = b_{n-1}^{(2)})
$$
where $b_{n-1}^{(1)}, b_{n-1}^{(2)} \in \{0,1\}$. Note that we have defined earlier $b_{n-1}^{(1)} \neq b_{n-1}^{(2)}$, which means 
$$
(q_{1} = b_{1} \wedge q_{2} = b_{2} \wedge q_{3} = b_{3} \cdots ) \wedge (q_{n-1} = 0 \vee q_{n-1} = 1)
$$
The expression $(q_{n-1} = 0 \vee q_{n-1} = 1)$ will always be true, so we are left with
$$
(q_{1} = b_{1} \wedge q_{2} = b_{2} \wedge q_{3} = b_{3} \wedge \cdots \wedge q_{n-2} = b_{n-2} \wedge q_{n} = b_{n})
$$
which is reduced to $C^{n-1} X$.
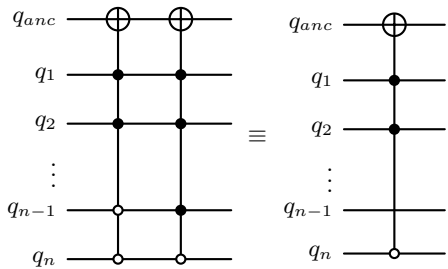
\begin{figure}[h!]
    \centering
    \begin{quantikz}
    \lstick{$q_{anc}$}& \targ{} & \targ{}& \\
    \lstick{$q_{1}$}& \ctrl{-1} & \ctrl{-1}& \\
    \lstick{$q_{2}$}& \ctrl{-1} & \ctrl{-1} &\\
    \lstick{\vdots}   \\
    \lstick{$q_{n-1}$}& \ctrl[open]{-2} & \ctrl{-2}& \\
    \lstick{$q_{n}$}& \ctrl[open]{-1} & \ctrl[open]{-1}&
\end{quantikz} $\equiv$
\centering
    \begin{quantikz}
    \lstick{$q_{anc}$}& \targ{} & \\
    \lstick{$q_{1}$}& \ctrl{-1} & \\
    \lstick{$q_{2}$}& \ctrl{-1} & \\
    \lstick{\vdots}   \\
    \lstick{$q_{n-1}$}&  & \\
    \lstick{$q_{n}$}& \ctrl[open]{-3} & 
\end{quantikz}
    \caption{Two $C^{n} X$ gates are combined into one $C^{n-1} X$.}
    \label{fig:placeholder}
\end{figure}
\end{proof}
Note that $U_{j,a}$ also contains two $C^{3} X$ gates, which we can reduce to one $C^{2} X$ according to Theorem \ref{theorem-3}. The binary representation of the non-zero rows of $\mathcal{H}_{j}\mathcal{H}_{j}^{T}$ yields the result that both $C^{3}X$ differ in $q_{3}$, which we can ignore to get one $C^{2}X$ (see Fig. \ref{example-1-circuit-construction}).
\end{example}
\begin{example}\label{U-construction-example-02}
    We now consider a (5+1) qubit system with $\mathcal{H}_{j} = \sigma_{+} \otimes \sigma_{+}\sigma_{-} \otimes I \otimes \sigma_{-} \otimes \sigma_{+}$, and the corresponding unitary completion $\Bar{\mathcal{H}_{j}} = \sigma_{x} \otimes I \otimes I \otimes \sigma_{x} \otimes \sigma_{x}.$ The $U_{j,b}$ will then require one extra $PauliX$ in addition to $\Bar{\mathcal{H}_{j}}$. For $U_{j,a}$, let first compute $\mathcal{H}_{j} \mathcal{H}_{j}^{T}$,
    $$
    \mathcal{H}_{j}\mathcal{H}_{j}^{T} = \begin{pmatrix}
        1 & 0 \\
        0 & 0
    \end{pmatrix} \otimes 
    \begin{pmatrix}
        1 & 0 \\
        0 & 0
    \end{pmatrix} \otimes
    \begin{pmatrix}
        1 & 0 \\
        0 & 1
    \end{pmatrix} \otimes
    \begin{pmatrix}
        0 & 0 \\
        0 & 1
    \end{pmatrix} \otimes
    \begin{pmatrix}
        1 & 0 \\
        0 & 0
    \end{pmatrix}
    $$
This will provide us a $32 \times 32$ matrix with two non-zero rows, as previous, having binary representation \{000010\} and \{000110\}. The notation is same as previous, that is, the most significant bit is ancilla. It will require two $C^{4} X$ gates that differ in the control operation on system qubit $q_{3}$, which we can ignore to form a single $C^{3} X$. The circuit for $U_{j}$ is given in Fig. \ref{fig.example-02}.
\begin{figure}[h!]
\centering
\begin{quantikz}
        \lstick{$\ket{q}_{anc}$} & \gate{X}\gategroup[6,steps=1,style={dashed,rounded
         corners,fill=blue!20, inner
         xsep=2pt},background,label style={label
         position=below,anchor=north,yshift=-0.2cm}]{{\sc
         $U_{j,b}$}}  & \targ{} \gategroup[6,steps=1,style={dashed,rounded
         corners,fill=blue!10, inner
         xsep=2pt},background,label style={label
         position=below,anchor=north,yshift=-0.2cm}]{{\sc
         $U_{j,a}$}} & \\
        \lstick{$q_{1}$} & \gate{X}  & \ctrl[open]{-1} & \\
        \lstick{$q_{2}$} & \gate{I} & \ctrl[open]{-1} & \\
        \lstick{$q_{3}$} & \gate{I} &  & \\
        \lstick{$q_{4}$} & \gate{X} & \ctrl{-2} & \\
        \lstick{$q_{5}$} & \gate{X} & \ctrl[open]{-1} &
    \end{quantikz}
    \caption{Circuit for $U_{j}$ corresponding to $\mathcal{H}_{j}$ defined in Example. \ref{U-construction-example-02}}
    \label{fig.example-02}
    \end{figure}
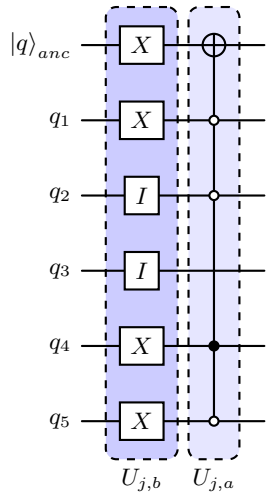
\end{example}
In the above examples, one could notice that, when we implement $U_{j,a}$, there is always an open control on the qubit if the corresponding term in $\mathcal{H}_{j}$ is $\sigma_{+}$ or $\sigma_{+}\sigma_{-}$. Similarly, a closed control appears on the qubit if the corresponding term satisfies $\mathcal{H}_{j}$ $\in \{\sigma_{-}, \sigma{-}\sigma{+}\}$. We can use this logic to implement $U_{j,a}$ for any $(n+1)$ qubit system. Let
$$
\mathcal{H}_{j} = \underbrace{\sigma_{+} \otimes \sigma_{-} \otimes \sigma_{+}\sigma_{-} \otimes \cdots \otimes \sigma_{-}\sigma_{+}\otimes \sigma_{+}}_{n \,\, terms} 
$$
\begin{figure}[h]
\centering
\begin{tikzpicture}
  \node[anchor=west] (circuit) at (0,0) {
    \begin{varwidth}{0.95\linewidth}
      \begin{quantikz}
        \lstick{$q_{anc}$} & \gate{X}\wire[d]{q}
          \gategroup[6,steps=1,style={dashed,rounded corners,fill=blue!20,inner xsep=2pt},background,
            label style={label position=below,anchor=north,yshift=-0.2cm}]{{\sc $U_{j,b}$}} & & 
          \targ{} 
          \gategroup[6,steps=1,style={dashed,rounded corners,fill=blue!10,inner xsep=2pt},background,
            label style={label position=below,anchor=north,yshift=-0.2cm}]{{\sc $U_{j,a}$}} & \\
        \lstick{$q_{1}$} & \gate[5]{\Bar{\mathcal{H}_{j}}}  & & \ctrl[open]{-1} & \\
        \lstick{$q_{2}$} &  & & \ctrl{-1}  & \\
        \lstick{\vdots} \\
        \lstick{$q_{n-1}$} &  & & \ctrl{1}  & \\
        \lstick{$q_{n}$} &  &  &\ctrl[open]{0} &
      \end{quantikz}
    \end{varwidth}
  };

  \node at ($(circuit.south east)+(-0.9cm,2.85cm)$) {\vdots};

\end{tikzpicture}

\caption{Circuit for $U_{j}$ corresponding to $\mathcal{H}_{j} = \otimes_{p} \sigma_{p}$ $\forall \sigma_{p} \in \mathbb{S}$, and $p = 1, \cdots n$.}
\label{figure-06}
\end{figure}
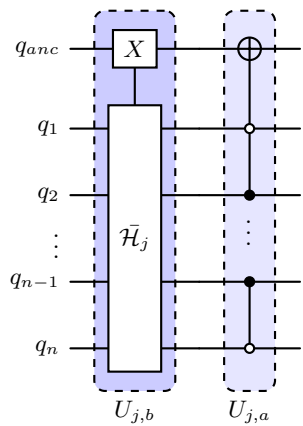

We have now investigated and built a formalism which led us to the implement of non-unitary operators using unitary completion. Our main focus was to implement $\mathcal{H}_{j}$, however, if we recall Eqn.(\ref{LCNU}) we need to implement the weighted sum of $\mathcal{H}_{j}$. As in Section. \ref{section-3.1}, we can define the block encoding for $U_{j}$. For that, we need ancillary qubits in the state $\ket{\psi}$ as described in Eqn.(\ref{ancilla-state}). The select operator for $U_{j}$, according to Eqn.(\ref{SEL-op}), will then take the form
\begin{equation}\label{SEL-op0}
    SEL\, \ket{0} \ket{i} \ket{\chi} = \ket{i} U_{j} \ket{0} \ket{\chi}
\end{equation}
where $\ket{\chi}$ is the state in which the system qubits are. Here, we leverage Theorem.4 and Lemma.7 from \cite{zhang2024circuit}, to construct the circuit for Eqn.(\ref{SEL-op0}), up to a certain depth of Clifford+T gates (\cite{zhang2024circuit}, see Lemma. 4). That is, the circuit for a select oracle for any general unitary function can be constructed based on the implementation of single-qubit controlled unitary gates. For the case in point, we re-define a general select oracle as 
$$
SEL (S) = \sum_{i=1}^{n} \ket{i} \ket{i} \otimes S_{i}
$$
where $S_{i} = \otimes_{j=1}^{n} S_{i,j}$ with $S_{i,j} \in \{I, \sigma_{x},\sigma_{y}, \sigma_{z}\}$. The authors have implemented each block $S_{i,j}$ from $SEL(S)$ as a controlled unitary. Fig. \ref{SEL-ORACLE} represents the circuit construction for controlled-$S_{i}$ using $n$ ancillae and target qubits. In a similar fashion, we can implement our unitary completion $U_{j}$ corresponding to the tensor product of non-unitaries $\mathcal{H}_{j}$. Consider $\mathcal{H}_{j,i}$ are the weighted tensor components of $\mathcal{H}$, such that
$$
\mathcal{H} = \alpha_{1} \mathcal{H}_{j,1} + \alpha_{2} \mathcal{H}_{j,2} + \cdots + \alpha_{n} \mathcal{H}_{j,n} \,\,\,\, (see\, Eqn.(\ref{LCNU}))
$$
According to Theorem. \ref{theorem-1}, the unitary completion for each block $\mathcal{H}_{j,i}$ will satisfy $\Bar{\mathcal{H}}_{j,i} \in \{I, \sigma_{x}\}$, that is, we can easily implement these completions in a similar way as we did for $S_{j,i}$. 
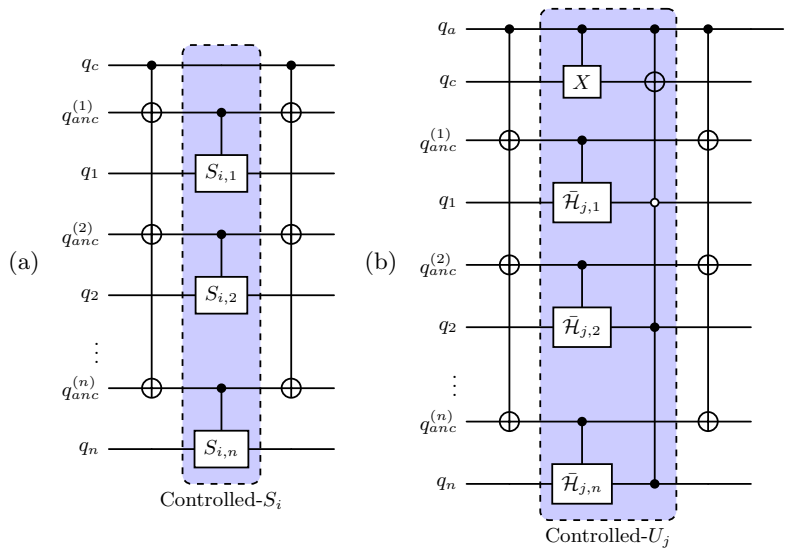
\begin{figure}[h!]
    \centering
    \begin{subfigure}
    {\small (a)}
        \centering
        \scalebox{0.85}{
        \begin{quantikz}
            \lstick{$q_{c}$}& \ctrl{6} & 
            \gategroup[8,steps=1,style={dashed,rounded
            corners,fill=blue!20, inner
            xsep=2pt},background,label style={label
            position=below,anchor=north,yshift=-0.2cm}]{{Controlled-$S_{i}$}}
            & \ctrl{6}& \\
            \lstick{$q_{anc}^{(1)}$}& \targ{} & \ctrl{1} & \targ{} & \\
            \lstick{$q_{1}$}& & \gate{S_{i,1}}& &  \\
            \lstick{$q_{anc}^{(2)}$}& \targ{} & \ctrl{1}& \targ{} &  \\
            \lstick{$q_{2}$}& & \gate{S_{i,2}} & &  \\
            \lstick{\vdots}\\
            \lstick{$q_{anc}^{(n)}$}& \targ{} & \ctrl{1}& \targ{} & \\
            \lstick{$q_{n}$} & & \gate{S_{i,n}} & &
        \end{quantikz}
        }
    \end{subfigure}
    \begin{subfigure}
        {\small (b)}
        \centering
        \scalebox{0.85}{
        \begin{quantikz}
            \lstick{$q_{a}$}& \ctrl{7} & \ctrl{1}
            \gategroup[9,steps=2,style={dashed,rounded
            corners,fill=blue!20, inner
            xsep=2pt},background,label style={label
            position=below,anchor=north,yshift=-0.2cm}]{{Controlled-$U_{j}$}}
            & \ctrl{1} & \ctrl{7} & & \\
            \lstick{$q_{c}$}&  & \gate{X} & \targ{}&  & \\
            \lstick{$q_{anc}^{(1)}$}& \targ{} & \ctrl{1} &  & \targ{} & \\
            \lstick{$q_{1}$}& & \gate{\Bar{\mathcal{H}}_{j,1}}& \ctrl[open]{-2}& &  \\
            \lstick{$q_{anc}^{(2)}$}& \targ{} & \ctrl{1}& & \targ{} &  \\
            \lstick{$q_{2}$}& & \gate{\Bar{\mathcal{H}}_{j,2}} & \ctrl{-2}& & \\
            \lstick{\vdots}\\
            \lstick{$q_{anc}^{(n)}$}& \targ{} & \ctrl{1}& & \targ{} &\\
            \lstick{$q_{n}$} & & \gate{\Bar{\mathcal{H}}_{j,n}} & \ctrl{-3}& &
        \end{quantikz}
        }
    \end{subfigure}
    \label{SEL-ORACLE}
    \caption{(a) Circuit representing a general SELECT oracle, where each weighted tensor component is a Pauli-string. $q_{c}$ is the controlled qubit, $q_{anc}^{(i)}$ are ancillae, and $q_{i}$ are target qubits. (b) Circuit representing SELECT oracle for the unitary completion operator $U_{j}$ corresponding to $\mathcal{H}_{j}$, where each weighted tensor component satisfies $\Bar{\mathcal{H}}_{j,i} \in \{I, \sigma_{x}\}$. Note that we need one extra ancilla $q_{a}$ for the implementation of $\Bar{\mathcal{H}}_{j,i}$.}
\end{figure}
\section{Curse of Dimensionality}\label{section-4}
Leveraging the Sigma basis to decrease the number of decomposition terms will result in the reduction of per-term implementation cost; however, it does not change the fact that the Carleman embedded system is an enormous linear system whose divergence requires a large value of the truncation order $N$. That is, though we have now an alternative decomposition basis that requires only one basis component for each non-zero entry of the matrix, yet the decomposition data is very large, which could eventually cause curse of dimensionality; a main reason for barren plateau phenomena in solving quantum systems using variational quantum computing algorithms. This phenomenon refers to a vanishing gradient problem in which the cost function, which is being minimized using the variational parameters, becomes flat, making the variational circuit untrainable. It first appeared in the domain of quantum neural networks\cite{mcclean2018barren}. 

In our case, the problem of interest is a Carleman embedded linear system that we aim to solve using variational quantum linear solvers algorithm. Most of the problems in this context are considered as an optimization task, through which, we quantify the quality of the solution using the minimization of a cost function through trainable parametrized quantum circuit. This is a hybrid approach in which we leverage quantum hardware for estimating the loss and then use classical optimizers to update the parameters. The parameter optimization becomes very hard in case of large Hilbert space and deeper circuits, which is precisely the case of a Carleman embedded linear system. Generally, the cost function takes the form
\begin{equation}\label{cost-function}
    l_\theta (O,\phi) = Tr\left[O U(\theta) \phi U(\theta)^{\dagger}\right]
\end{equation}
where $O$ is an observable, $\phi$ is the initial state of the system, and $U$ is any parametrization that needs to be optimized by updating $\theta$. In classical optimization, our aim is to achieve $\Bar{\theta}$ such that
\begin{equation}\label{optimized-param}
    \Bar{\theta} = arg_{\theta} \, \text{min}\, \mathcal{C}(\theta)
\end{equation}
To reach $\Bar{\theta}$, gradient-based methods have been widely used. One such method is the gradient-decent algorithm,
\begin{equation}\label{gradient-decent}
    \theta^{t+1} = \theta^{t} - \beta \nabla_{\theta^{t}} \mathcal{C}(\theta^{t})
\end{equation}
where $\beta$ is the learning rate. The more formal and mathematical definition of the barren plateau can be obtained using Chebyshev's inequality\cite{friedrich2025barren},
\begin{equation}\label{Chebyshev-inequality}
    \mathbb{P} \left(\left|\partial_{\mu} \mathcal{C} - \langle \partial_{\mu}\mathcal{C} \rangle\right|\geq \eta\right) \leq \frac{\mathbb{V}[\langle \partial_{\mu} \mathcal{C}\rangle]}{\eta^{2}}
\end{equation}
which represents an upper bound on the probability of the deviation of the partial derivative of $\mathcal{C}$, with respect to the variational parameter $\theta_{\mu}$, from its mean value $\langle \partial_{\mu} \mathcal{C} \rangle$ by a value greater than or equal to $\eta$. This inequality followed by $\langle \partial_{\mu} \mathcal{C} \rangle = 0$ \cite{cerezo2021cost,mcclean2018barren} will result in the original definition of the problem:
\begin{equation}\label{barren-plateau}
    \mathbb{V}\left[\langle \partial_{\mu} \mathcal{C} \rangle\right] \leq \mathcal{B}(n) \propto \mathcal{O}\left(\frac{1}{a^{n}}\right), a>1
\end{equation}
This states that the distribution will now start squeezing tightly around 0 and will just provide a narrow spike. That is, for small values of the variance, the gradient will hardly move away from 0. One unique way to address this problem lies in the choice of a cost function. Generally, any observable $O$ from Eqn.(\ref{cost-function}) can be described in two different ways (a) local observable (the corresponding cost function will then be referred to as the local cost function) and (b) global observable $O$ (the corresponding cost function will then be the global cost function), from which the former is advantageous over the latter, as proved by \cite{cerezo2021cost}. The global observable is the one in which the corresponding cost function depends on the simultaneous measurement on all of the qubits.  Such an observable is highly sensitive to parameter change and once it is affected, the information will flow throughout the circuit. In deep and randomly initialized variational circuits where state of the system becomes scrambled, these observables mix the information from many qubits by averaging out in an extremely strong way. As a result, an exponential decay can be observed in the variance of the gradient as the number of qubits increases, making classical optimizers unable to detect the tiny meaningful information. This exponential decay is responsible for making the circuit untrainable and the optimization landscape flat. The local cost function, on the other hand, is the one in which the observable's support remains manageable or constant, i.e., it acts non-trivially on a small number of qubits, no matter how large the system is. In this way, the gradient of the variational circuit having a local cost function will only depend on the variational parameters that affect the set of qubits that are being measured, or more generally, the parameters lying inside the light-cone generated by the measured qubits. A light-cone is referred to as the region of the circuit that is connected with the qubit which we are going to measure (see Fig. \ref{figure-9}a). No matter how many working qubits do we have, this light-cone will remain the same in size. This will create a bound region, in which the variance of the gradient will remain at least polynomially large, which guaranties the trainability of the circuit.
\begin{figure}[h!]
    \centering
	\begin{minipage}{0.4\textwidth}
        \centering
        {\small (a)}
        \includegraphics[width=\linewidth]{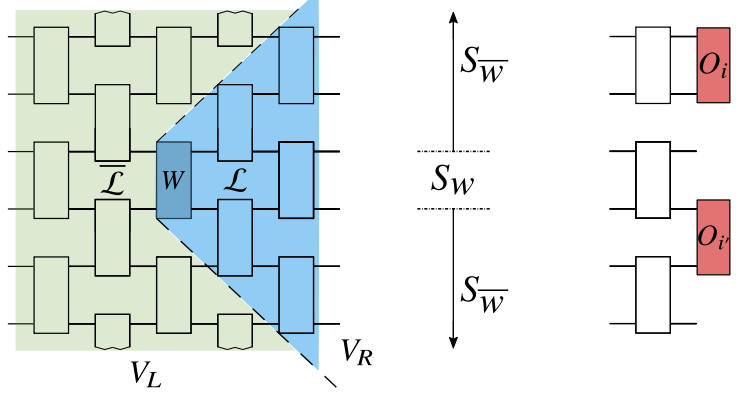}
    \end{minipage}
    \begin{minipage}{0.25\textwidth}
        {\small (b)}
        \centering
        \includegraphics[width=\linewidth]{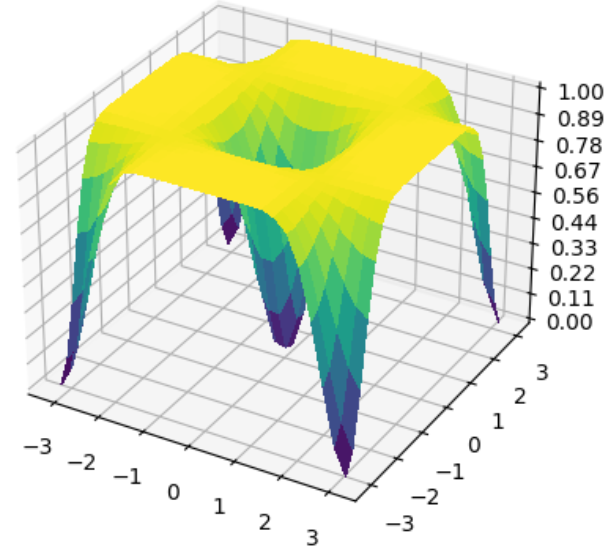}
    \end{minipage}
    \begin{minipage}{0.25\textwidth}
        \centering
        {\small (c)}
        \includegraphics[width=\linewidth]{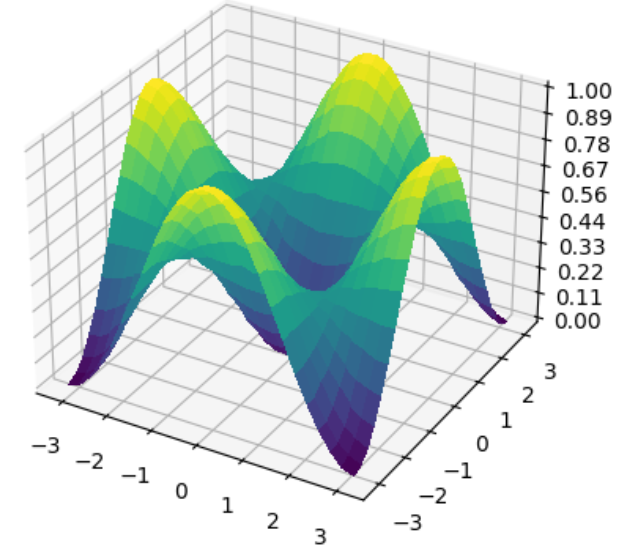}
    \end{minipage}
    \begin{minipage}{1\textwidth}
        \centering
        \includegraphics[width=\linewidth]{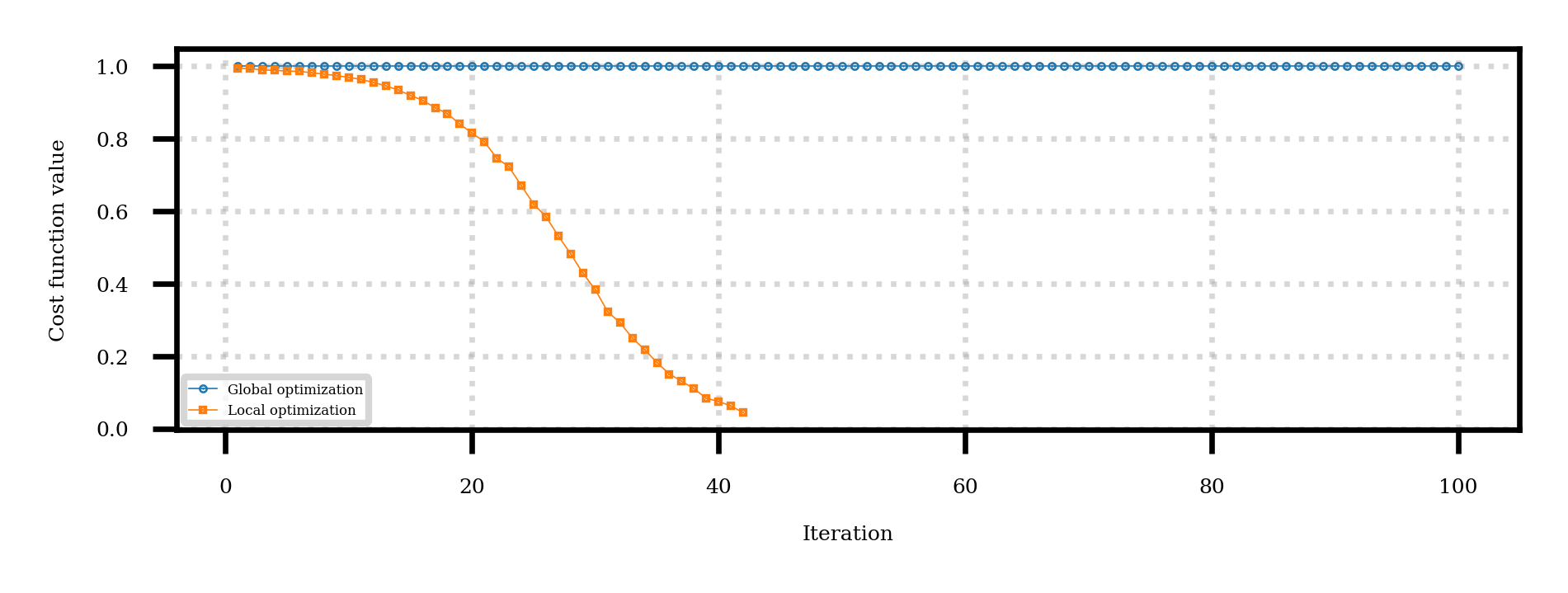}
    \end{minipage}
      \caption{(a) A forward light-cone $\mathcal{L}$ spanning the gates $V_{R}$ corresponding to the unitary block $W$. $S_{w}$ is the subset of the qubits of interest or those which are being measured by a local observable $O_{i}$. (This image is taken from \cite{cerezo2021cost}).  (b) A difficult to train global cost function plotted across 6 qubits, showing much of the landscape flat. The effect gets worse with increase in the number of qubits. (c) Local cost function across same number of qubits and ansatz, representing the trainability of the circuit and decreased effect of barren plateaus. (Both (b) and (c) are taken from \cite{ThomasStorwick2020}). The bottom plot is the comparison between global and local cost function minimization. The global cost is stuck to $1$ making the variational circuit untrainable on new set of experiments with new variational parameters while the local cost function shows the minimization as iterations increases. (\textit{source code: \url{https://github.com/ali-tayyab/Reducing-Barren-Plateau-using-Local-Cost-Functions}})}
   \label{figure-9}
\end{figure}

The visual demonstration of this problem and the effect of shifting from global to local cost function have been shown in \cite{ThomasStorwick2020} (a pennylane's demo based on \cite{cerezo2021cost}). The author used a hardware efficient ansatz in which the circuit is built using repeated rotations along $X$ and $Y$ on a 6 qubit system with multiple repeated blocks to scramble the state of the system. We leverage this work to plot local and global cost versus the number of iterations.
\section{Conclusion}\label{section-5}
Linearity violation in quantum mechanics could have dramatic operational consequences in leveraging quantum computers to solve nonlinear problems. We need to re-engineer the current quantum hardware on the principles of nonlinear variants of quantum mechanics in order to solve a nonlinear system on it directly. To avoid this long term challenge, researchers in the field have been working on the linear representations of a nonlinear system using linearization frameworks. These frameworks are very useful for utilizing the quantum advantage in solving nonlinear systems, since current quantum computers are capable of solving linear systems only. In this study, we use Carleman linearization to address the decomposition problem for a truncated linear system using an efficient decomposition model of non-unitary operators.  To illustrate the exponential reduction in the decomposition terms of the Hamiltonian, we use a python script (\textit{\url{https://github.com/ali-tayyab/Carleman-Embedding}}) to decompose sparse matrices containing different nonzero patterns into the weighted sum of the Sigma basis. However, numerical and semi-analytical techniques have also been studied in \cite{gnanasekaran2024efficient}. We compare this decomposition with the traditional decomposition of the Hamiltonian into a linear combination of Pauli operators. 

No matter how efficiently we can decompose a Hamiltonian into the weighted sum of any arbitrary basis, if the underlying basis set contains operators that cannot be implemented on a quantum computer, all is vein. For this reason, it is quite obvious for any reader to mark a question on the non-unitarity of the components of $\mathbb{S}$. We answer this question by constructing a general circuit to compute each weighted tensor product component of the decomposed Hamiltonian using unitary completion. 

The large number of terms in the Hamiltonian of any real-world physical system constrains the scalability of the established Hamiltonian simulation methods. For this, researchers worked on the sparsification of the actual Hamiltonian to reduce the number of terms\cite{ouyang2020compilation}. Such methods reduce the computational cost at the expense of exactness, which is critical in several sensitive quantum systems. Exploring these methods and comparing the circuit depth along with the solution accuracy with our proposed model is an important direction for future work. \newpage
{\bf Acknowledgements.}
The author would like to express sincere gratitude to Abeynaya Gnanasekaran and Amit Surana for their valuable lectures during the WISER Quantum Program 2025. The author also thanks the organizers of this summer program for providing an excellent learning environment.
\bibliography{citations}
\bibliographystyle{unsrt}
\end{document}